\documentclass[review]{elsarticle}
\usepackage[left=100pt, right=100pt]{geometry}
\usepackage{lineno,hyperref}
\modulolinenumbers[5]
\usepackage[utf8]{inputenc}
\usepackage{graphicx}
\usepackage{mwe}
\usepackage{subfig}
\usepackage {float}
\usepackage{amssymb}
\usepackage{amsmath}
\usepackage{dsfont}
\usepackage[english]{babel}
\usepackage{amssymb}
\usepackage{mathrsfs}
\usepackage{frcursive}
\usepackage{mathtools}
\usepackage{amsthm}
\usepackage[utf8]{inputenc}
\usepackage[T1]{fontenc}
\usepackage{lmodern}
\usepackage{ae}
\usepackage{enumerate}
\newtheorem{theorem}{Theorem}[section]
\newtheorem{lemma}{Lemma}

\newtheorem{definition}{Definition}
\newtheorem{example}{Example}
\journal{Journal of \LaTeX\ Templates}

\newcommand{\Rb}{\mathbb{R}}

\begin{document}

\begin{frontmatter}
	
	\title{\textbf{Applied Probability Insights into Nonlinear Epidemic Dynamics with Independent Jumps}}
	\author[label1]{ Brahim Boukanjime\corref{}}

\ead{brahim.boukanjime@gmail.com}
\affiliation[label1]{organization={Laboratory of Mathematical Modeling and Economic Calculations, Hassan First University of Settat},
  country={Morocco}}

\author[label2]{ Mohamed El Fatini\corref{}}

\ead{melfatini@gmail.com}
\affiliation[label2]{organization={Ibn Tofail University, Faculty of Sciences, Department of Mathematics, BP 133, Kénitra},
  country={Morocco}}

\author[label3]{Mohamed Maama\corref{cor}}
\cortext[cor]{Corresponding author}
\affiliation[label3]{organization={Applied Mathematics and Computational Science Program,  Computer, Electrical and Mathematical Sciences and Engineering Division, King Abdullah University of Science and Technology},
	country={Kingdom of Saudi Arabia}}
 \ead{maama.mohamed@gmail.com}

\begin{abstract}
This paper focuses on the analysis of a stochastic SAIRS-type epidemic model that explicitly incorporates the roles of asymptomatic and symptomatic infectious individuals in disease transmission dynamics. Asymptomatic carriers, often undetected due to the lack of symptoms, play a crucial role in the spread of many communicable diseases, including COVID-19. Our model also accounts for vaccination and considers the stochastic effects of environmental and population-level randomness using Lévy processes. We begin by demonstrating the existence and uniqueness of a global positive solution to the proposed stochastic system, ensuring the model's mathematical validity. Subsequently, we derive sufficient conditions under which the disease either becomes extinct or persists over time, depending on the parameters and initial conditions. The analysis highlights the influence of random perturbations, asymptomatic transmission, and vaccination strategies on disease dynamics. Finally, we conduct comprehensive numerical simulations to validate the theoretical findings and illustrate the behavior of the model under various scenarios of randomness and parameter settings. These results provide valuable insights into the stochastic dynamics of epidemic outbreaks and inform strategies for disease management and control.

\end{abstract}
\begin{keyword}
SAIRS epidemic model, Lévy processes, Asymptomatic transmission, Stochastic extinction, Persistence in mean, Vaccination dynamics.
\end{keyword}
	
\end{frontmatter}

\section{Introduction}

Infectious diseases  spread geographically on a large scale in a short period of time, affecting the lives of people around the world, such as the current COVID-19 pneumonia\cite{1,2,3,4}. For example, according to the World Health Organization (WHO), globally, as of 3:05 p.m. CET, December 20, 2020, there have been 75,098,369 confirmed cases of COVID-19, including 1,680,339 deaths, reported to WHO \cite{who}. Currently, we are struggling to control the spread of this epidemic, which has the potential to be the deadliest disease in human history. In this regard, to face and manage such a crisis situation, various scientific methods and analytical studies are required. As a method and science, mathematical modeling is the appropriate tool. It is the art of translating real problems into flexible mathematical formulas whose theoretical and numerical analysis provides suggestions,
answers, and constructive solutions. As a field of application, mathematical epidemiology plays the main role in the analysis of factors that may influence the prevalence of the disease.
\\
A peculiar, yet crucial feature of the recent Covid-19 pandemic is the existence of an asymptomatic state, where an infected individual shows no symptoms but can potentially trasnmit the infection to other individuals. This is one of the main aspect that has allowed the virus to circulate widely in the population, since asymptomatic cases often remain unidentified, and presumably have more contacts than symptomatic cases, since lack of symptoms often implies a lack of quarantine. Hence, the contribution of the so called “silent spreaders” to the infection transmission dynamics are relevant for various communicable diseases, such as Covid-19, influenza, cholera and shigella \cite{12,13,14,15}.
\\ 
Models that incorporate an asymptomatic compartment already exist in literature, but have not been analytically studied as thoroughly as more famous compartmental models. For example, in \cite{rob}, the authors studied a mathematical model which describes the dynamics of an aerially transmitted disease, under explicit consideration of asymptomatic cases. In the same context, Ansumali et al. \cite{anso} discussed the dynamics of SARS-CoV-2 pandemic with asymptomatic patients. In another work, Ottaviano et al. \cite{ottaw}, has extended the result of \cite{rob}, in the case of vaccination and found a threshold condition between persistence and extinction. Actually, the paper \cite{global} is considered to be an extension of the study \cite{ottaw}, which presented a model that includes varying population, vaccination of newborns, disease induced deaths and treatment to symptomatic individuals. 
\\
As far as we know, natural phenomenon is always affected by environmental factors that can aggravate or mitigate the spread of the epidemic. The stochastic quantification of several real-life phenomena has been immensely helpful in understanding the random nature of their incidence or occurrence\cite{5,6,7,8,10}. It has also helped to find solutions to those problems that arise from it, either in the form of minimizing their undesirability or maximizing their rewards. In such cases, deterministic systems, while able to make very informative forecasts and previsions, are not appropriate enough \cite{doo}. So, there is a pressing need for a developed mathematical model that can take into account the randomness effect.
 Since the inclusion of noises of Lévy provides more mutual information or good results for various infectious diseases in the form of stochastic models. Then, it will be more reasonable to illustrate those sudden fluctuations into the infection model. For this reason, we will make recourse to the well known jump Lévy processes. By taking into account this type of random perturbations, we extend the  deterministic SAIRS (Susceptible-Asymptomatic infected-symptomatic Infected-Recovered-Susceptible) model proposed in \cite{global}  to the following modified system  
\begin{eqnarray}\label{sys}
	\left\{ \begin{array}{ll}
		dS(t)=\bigg[b(1-\nu)-\dfrac{\beta_ASA}{N}-\dfrac{\beta_ISI}{N}-(\rho+\mu)S+\xi R\bigg]dt+S(t^-)d\mathcal{Z}_1(t),\\
		dA(t)=\bigg[\dfrac{\beta_ASA}{N}+\dfrac{\beta_ISI}{N}-(\gamma_A+\sigma+\mu)A\bigg]dt+A(t^-)d\mathcal{Z}_2(t),\\
		dI(t)=[\sigma A-(\gamma_I+\eta+\alpha+\mu)I]dt+I(t^-)d\mathcal{Z}_3(t),\\
		dR(t)=[b\nu+\rho S+\gamma_AA+(\gamma_I+\eta)I-(\xi+\mu)R]dt+R(t^-)d\mathcal{Z}_4(t)
	\end{array} \right.,
\end{eqnarray}
where the total population $N$ is partitioned into four compartments, namely $S$, $A$, $I$, and $R$, which represent the fraction of susceptible, asymptomatic infected, symptomatic infected and recovered individuals, respectively, such that $N = S + A + I + R$. The infection can be transmitted to a susceptible through a contact with either an asymptomatic infectious individual, at rate
$\beta_A$, or a symptomatic individual, at rate  $\beta_I$. Once infected, all susceptible individuals enter an asymptomatic state, indicating a delay between infection and symptom onset if they occur. Indeed, we include in the asymptomatic class both individuals who will never develop the symptoms and pre-symptomatic who will eventually become symptomatic. From the asymptomatic compartment, an individual can either progress to the class of symptomatic infectious $I$, at rate $\sigma$, or recover without ever developing symptoms, at rate $\gamma_A$. An infected individuals with symptoms can recover at a rate $\gamma_I$ or are given treatment at rate $\eta$. Disease-related deaths are at rate $\alpha$. We assume that the recovered individuals do not obtain a long-life immunity and can return to the susceptible state after an average time $1/\xi$. We
also assume that a proportion $\nu$ of susceptible individuals receive a dose of vaccine which grants them a temporary immunity. We do not add a compartment for the vaccinated individuals, not distinguishing the vaccine-induced immunity from the natural one acquired after recovery from the virus. We consider the vital dynamics of the entire population and, we assume that the rate of births is $b$ and natural death rate is $\mu$. Further, they assume a proportion of newborns, $\nu$ are vaccinated who enter the recovered class while the remaining enter the susceptible class. Moreover, susceptibles are vaccinated at rate $\rho$. Here and elsewhere,
$S (t^-)$, $A(t^-)$, $I(t^-)$ and $R(t^-)$ are respectively the
left limits of the functions $S(t)$, $A(t)$, $I(t)$ and $R(t)$. The vector $\mathcal{Z}(t) = (\mathcal{Z}_1(t), \mathcal{Z}_2(t), \mathcal{Z}_3(t),\mathcal{Z}_4(t))$ indicates the $4$-dimensional Lévy process with its associated Lévy-Khintchine formula: 
\begin{eqnarray*}
\mathbb{E}\bigg\{e^{ik_1\mathcal{Z}_1(t)+ik_2\mathcal{Z}_2(t)+ik_3 \mathcal{Z}_3(t)+ik_4\mathcal{Z}_4(t)}\bigg\}=\exp\bigg\{-\frac{t}{2}<k,Qk > + t\displaystyle\int_{\mathbb{R}\backslash 0}(e^{i<k,q(z)>}-i<k,q(z)>-1)\nu(dz) \bigg\},
\end{eqnarray*}
where $\mathbb{E}$ is the mathematical expectation,  $k = (k_1, k_2, k_3,k_4)\in \mathbb{R}^4$, $t\in\mathbb{R}_+,$ and $Q$ is represents a represents a specific positive definite matrix, $B_i(t)~~(i = 1, 2, 3,4)$ are mutually independent standard Brownian motions with $B_i(0) = 0$ defined on a complete probability space $(\Omega,\mathcal{F}, P)$
	with a filtration $\{\mathcal{F}_t\}_{t\geq 0}$ which satisfies the usual conditions. $\sigma_i~~(i = 1, 2, 3,4)$ represent the intensities of the Gaussian white noise. The compensated Poisson random measure is represented by $\tilde{N}(dt, dz) := \mathcal{N}(dt, dz) - \nu(dz)dt$. $\mathcal{N}(dt, dz)$ is the Poisson random measure with characteristic measure $\nu(dz)$ on a measurable subset $\mathbb{Z}$ of $(0, \infty)$ which is a finite Lévy measure, such that $\displaystyle \int_{\mathbb{Z}}\min\{1,\lvert q_i(z)\rvert^2\}\nu(dz)<\infty~~(i = 1, 2, 3,4)$. The jumps amplitude $q_i: \mathbb{Z}\rightarrow\mathbb{R}~~(i = 1, 2, 3,4)$
are bounded and continuously differentiable.\\
In accordance with the theory exhibited in \cite{levy} and \cite{ley}, the process $\mathcal{Z}$ takes the following form:
\begin{eqnarray*}
\mathcal{Z}_i(t)=\sigma_i B_i(t)+\displaystyle\int_0^t \int_{\mathbb{Z}}q_i(z)\,\mathrm{{\tilde{N}}}(ds,dz)~~i=1,2,3,4.
\end{eqnarray*}
Based on the above assumptions, the model \eqref{sys} can be described by
\begin{eqnarray}\label{sys1}
	\left\{ \begin{array}{ll}
	dS(t)=\bigg[b(1-\nu)-\dfrac{\beta_ASA}{N}-\dfrac{\beta_ISI}{N}-(\rho+\mu)S+\xi R\bigg]dt+\sigma_1 SdB_1+\displaystyle \int_{\mathbb{Z}}q_1(z)S(s-)\,\mathrm{{\tilde{N}}}(ds,dz),\\\
	dA(t)=\bigg[\dfrac{\beta_ASA}{N}+\dfrac{\beta_ISI}{N}-(\gamma_A+\sigma+\mu)A\bigg]dt+\sigma_2 A dB_2+\displaystyle \int_{\mathbb{Z}}q_2(z)A(s-)\,\mathrm{{\tilde{N}}}(ds,dz),\\
	dI(t)=[\sigma A-(\gamma_I+\eta+\alpha+\mu)I]dt+\sigma_3 I dB_3+\displaystyle \int_{\mathbb{Z}}q_3(z)I(s-)\,\mathrm{{\tilde{N}}}(ds,dz),\\
	dR(t)=[b\nu+\rho S+\gamma_AA+(\gamma_I+\eta)I-(\xi+\mu)R]dt+\sigma_4 R dB_4+\displaystyle \int_{\mathbb{Z}}q_4(z)R(s-)\,\mathrm{{\tilde{N}}}(ds,dz)
	\end{array} \right.
\end{eqnarray}
In this paper, we analyze and investigate the behavior of stochastic SAIRS-type model with vaccination, where the role of asymptomatic and symptomatic infectious individuals is explicitly considered in the epidemic dynamics which it is performed by a Lévy process.  This original idea extends the studies presented in [our] and gives us a general view of the disease dynamics under different scenarios of random perturbations. 
\\
The outline of the work is the following. In Section 2,  we firstly prove that there is a unique positive global solution of the stochastic SAIRS model \eqref{sys1}. Then, we investigate the sufficient conditions for the disease
extinction and persistence. In section 3: Some numerical solutions are presented to validate the obtained analytical findings.

\section{Main theoretical results}
Before exhibiting the pivotal outcome of this work, it is perhaps the most significant step to show if the model \eqref{sys1} admits a solution, and if this solution is unique, positive and global (in time). In what follows, we will provide some assumptions under which the well-posedness of the SAIRS model \eqref{sys1} is ensured. But before doing so, let us first introduce  the following fundamental assumptions on the jump-diffusion coefficients.

\begin{eqnarray*}
	&(\textbf{A1})&~\mbox{For each}~N>0,~\mbox{there exists}~L_N>0~\mbox{such that}\\&&
	\displaystyle \int_{\mathbb{Z}}\mid \mathcal{J}_i(x_1,z)-\mathcal{J}_i(x_2,z)\mid ^2 \, \mathrm{}\nu (dz)\leqslant L_N\mid x_1-x_2\mid ^2~~i=1,2,3,4, ~ \mbox{with}~\mid x_1\mid\vee \mid x_2\mid\leqslant N,
\end{eqnarray*}
where
\begin{eqnarray*}
	\mathcal{J}_1(x,z)&=&q_1(z)x~\mbox{for}~x=S(t-),\\
\mathcal{J}_2(x,z)&=&q_2(z)x~\mbox{for}~x=A(t-),\\
\mathcal{J}_3(x,z)&=&q_3(z)x~\mbox{for}~x=I(t-),\\
\mathcal{J}_4(x,z)&=&q_4(z)x~\mbox{for}~x=R(t-).
\end{eqnarray*}
\begin{align*}
	(\textbf{A2})~~1+q_i(z)>0,~z\in\mathbb{Z},~i=1,2,3,4~\mbox{and there exist a postive constant}~C~\mbox{such that}
\end{align*}
\begin{eqnarray*}
	\mid \log(1+q_i(z))\mid\leqslant C.
\end{eqnarray*}
For simplicity, it will be practical to use the following notations throughout the rest of the paper:
\begin{eqnarray*}
\tau_1&=&\dfrac{\mu}{\bigg(\rho+\mu+\dfrac{\sigma_1^2}{2}+\displaystyle \int_{\mathbb{Z}}[q_1(z)-\log(1+q_1(z))]\, \mathrm{}\nu (dz)\bigg)},\\
\tau_2&=&\bigg(\gamma_A+\sigma+\mu +\dfrac{\sigma_2^2}{2}+\displaystyle\int_{\mathbb{Z}}[q_2(z)-\log(1+q_2(z))]\, \mathrm{}\nu (dz)\bigg),\\
\kappa_1&=&\bigg(\rho+\mu+\dfrac{\sigma_1^2}{2}+\displaystyle \int_{\mathbb{Z}}[q_1(z)-\log(1+q_1(z))]\, \mathrm{}\nu (dz)\bigg)\\
\kappa_2&=&\bigg(\gamma_A+\sigma+\mu +\dfrac{\sigma_2^2}{2}+\displaystyle \int_{\mathbb{Z}}[q_2(z)-\log(1+q_2(z))]\, \mathrm{}\nu (dz)\bigg)\\
\mathcal{T}_s&=&\dfrac{b(\beta_A+\beta_I)}{\mu\bigg[\gamma_A+\sigma+\mu+\dfrac{\sigma_2^2}{2}+\displaystyle\int_{\mathbb{Z}}[q_2(z)-\log(1+q_2(z)]\,\mathrm{\nu}(dz)\bigg]},\\
\mathcal{T}_s^*&=&\dfrac{\beta_Ab(1-\nu)}{\kappa_1\kappa_2},\\
\langle x(t)\rangle &=&\displaystyle\int_{0}^t x(s)ds.
\end{eqnarray*}

\begin{theorem}
Under $(\textbf{A$1$})$ and $(\textbf{A$2$})$, for any given initial value $(S_0,A_0,I_0,R_0)$, there is a unique
solution $(S(t),A(t),I(t),R(t))$ of model \eqref{sys1} defined on $t\geqslant 0$, and will remain in $\mathbb{R}^4_+$ with probability one.
\end{theorem}
\begin{proof}
Due to the local Lipschitz condition of coefficients of system \eqref{sys1}, there exists a unique local solution ($S(t),A(t),I(t),R(t)$) on $t\in[0,\tau_e),$ ($\tau_e$ is the explosion time) for any starting  value ($S(0),A(0),I(0),R(0)$).
Let $m_0>1$ be sufficiently large for ($S(t),A(t),I(t),R(t)$) to enter in the interval $[\frac{1}{m_0}, m_0]$. For each integer $m\geqslant m_0$, we consider the following stopping time
\begin{equation}
	\tau_m = \inf \left \{t \in [0, \tau_e ) : S(t)\not\in (\frac{1}{m},m) \;\text{or}\;  A(t) \not\in (\frac{1}{m},m)\;\text{or}\;  I(t) \not\in (\frac{1}{m},m)\;\text{or}\;  R(t) \not\in (\frac{1}{m},m) \right\}.
\end{equation}
We set $\inf \emptyset = \infty$ ($\emptyset$ denotes the empty set). Obviously $\tau_m\leqslant \tau_e$ a.s., if $\tau_{\infty}=\infty$ a.s. is true, then $\tau_{e}=\infty$ a.s. which mean that $(S(t),A(t) ,I(t),R(t))\in \mathbb{R}^4_+$ a.s. for all $t \geqslant -\tau.$\\
We will proceed by contradiction. Assume that $\tau_m<\infty,$  then there exists a pair of constant $T > 0$ and $\varepsilon \in( 0 , 1 )$ such that
\begin{center}
	$\mathbb{P}\{\tau_{\infty} \leqslant T\}> \varepsilon$.
\end{center}
Therefore, there is an integer $m_1\geqslant m_0$ such that
\begin{center}
	$\mathbb{P}\{\tau_m  \leq T\}\geqslant \varepsilon$ for all $m\geqslant m_1$. 
\end{center}
Let a $C^2$-function $V: \Rb^4_+ \mapsto \Rb_+$ by
$$V(S,A,I,R)=(S-1-\log S)+(A-1-\log S)+(I-1-\log I) +(R-1-\log R),$$
The nonnegativity of this function can be seen from $ u-1-\log u \geq 0.$ \\
Making use It\^o's formula to $V$, we get

\begin{eqnarray}\label{ex}
	dV&=&LVdt+\sigma_1(S-1)dB_1+\sigma_2(A-1)dB_2+\sigma_3(I-1)dB_3+\sigma_4(R-1)dB_4\nonumber \\&&\nonumber
	+\displaystyle \int_{Z} [q_1(z)S-\log(1+q_1(z))] \, \mathrm{{\tilde{N}}}(ds,dz)+\displaystyle \int_{Z} [q_2(z)A-\log(1+q_2(z))] \, \mathrm{{\tilde{N}}}(ds,dz)\\&&+\displaystyle \int_{Z} [q_3(z)I-\log(1+q_3(z))] \, \mathrm{{\tilde{N}}}(ds,dz)+\displaystyle \int_{Z} [q_4(z)R-\log(1+q_4(z))] \, \mathrm{{\tilde{N}}}(ds,dz),
\end{eqnarray}
By using the condition $(\textbf{A2})$, we get
\begin{eqnarray*}
LV&\leq & b+ \xi +\beta_A+\beta_I+4\mu+\gamma_A+\gamma_I+\sigma+\eta+\alpha+\rho+\dfrac{\sigma_{1}^2+\sigma_{2}^2+\sigma_3^2+\sigma_4^2}{2}+\displaystyle\int_{\mathbb{Z}}[q_1(z)-\log(1+q_1(z)]\,\mathrm{\nu}(dz)\\&&
+\displaystyle\int_{\mathbb{Z}}[q_2(z)-\log(1+q_2(z)]\,\mathrm{\nu}(dz)+\displaystyle\int_{\mathbb{Z}}[q_3(z)-\log(1+q_3(z)]\,\mathrm{\nu}(dz)+\displaystyle\int_{\mathbb{Z}}[q_4(z)-\log(1+q_4(z)]\,\mathrm{\nu}(dz):=K
\end{eqnarray*}

Integrating both sides of \eqref{ex} between $0$ and $\tau_m \wedge T$ and taking expectation we get
\begin{eqnarray*}
	0&\leqslant & \mathbb{E} V(S(\tau_m \wedge T),A(\tau_m \wedge T), I(\tau_m \wedge T),R(\tau_m \wedge T))\\
	&\leq & V(S(0),A(0),I(0),R(0))+KT.
\end{eqnarray*}
Define for each $u>0, ~U(u):=\inf\lbrace V(x_1,x_2,x_3,x_4), x_i\geqslant u~~or~~x_i\leqslant \frac{1}{u},i=1,2,3,4\rbrace$, $x_1=S,~x_2=A,~x_2=I,~x_2=R,$ and we have $\lim_{u \rightarrow \infty} U(u)=\infty$.\\
Therefore
\begin{eqnarray*}
	V(S(0),A(0),I(0),R(0))+KT &\geqslant & \mathbb{E} [\mathds{1}_{\{\tau_m\leqslant T\}} V(S(\tau_m \wedge T),A(\tau_m \wedge T), I(\tau_m \wedge T),R(\tau_m \wedge T))]
	\\
	&\geqslant & \epsilon U(m).
\end{eqnarray*}
Letting $m\longrightarrow \infty$ leads to $\infty >	V(S(0),A(0),I(0),R(0))+KT=\infty $ which is a contradiction. Hence $\tau_{\infty}= \infty$ a.s. and the model has a unique global solution  $(S(t),A(t),I(t),R(t))\in \mathbb{R}^4_+$ a.s.
\end{proof}

\subsection{Stochastic extinction}
In mathematical epidemiology, our prime concern after proving the well-posedness is to know if the disease will disappear or it will continue to exist. In this subsection, we investigate the necessary conditions for the the dies outing of disease through stochastic approach modeling. \\
Before stating the main result of this section, we must firstly give the following useful lemmas:
\begin{lemma}\cite{38}
Let $(S(t),A(t),I(t),R(t))$ be the positive solution of system \eqref{sys1} with any given initial condition $(S(0),A(0),I(0),R(0))\in\mathbb{R}_+^4$. Then 
$\lim\limits_{t\rightarrow \infty}\dfrac{S(t)+A(t)+I(t)+R(t)}{t}=0$ a.s., further, if $\mu>\dfrac{\sigma_{1}^2\vee\sigma_{2}^2\vee\sigma_{3}^2\vee\sigma_4^2}{2}$, then
\begin{align*}
\lim\limits_{t\rightarrow \infty}\displaystyle\int_{0}^t S(s)dB_1(s) = 0,\\
\lim\limits_{t\rightarrow \infty}\displaystyle\int_{0}^t A(s)dB_2(s) = 0,\\
\lim\limits_{t\rightarrow \infty}\displaystyle\int_{0}^t I(s)dB_3(s) = 0,\\
\lim\limits_{t\rightarrow \infty}\displaystyle\int_{0}^t R(s)dB_4(s) = 0,\\
\end{align*}
\end{lemma}
\begin{proof}
 The proof of this lemma is similar in spirit to that of lemmas 2.1 and 2.2 of \cite{38} and therefore it is omitted here.
\end{proof}
\begin{lemma}\label{strg}
	If a function which is continuous $\mathcal{M}=\{\mathcal{M}\}_{t\geq 0}$  exist at a local martingale such that at $t\rightarrow 0$ it dies then 
	\begin{eqnarray*}
		\lim_{t\rightarrow \infty}\langle \mathcal{M},\mathcal{M}\rangle_t=\infty, ~~\mbox{a.s.,}~~\Rightarrow~~\lim_{t\rightarrow \infty}\dfrac{\mathcal{M}_{t}}{\langle \mathcal{M},\mathcal{M} \rangle_t}=0,~~\mbox{a.s.}
	\end{eqnarray*}
	\begin{eqnarray*}
		\limsup_{t\rightarrow \infty}\dfrac{\langle \mathcal{M},\mathcal{M} \rangle_t}{t}<0,~~\mbox{a.s.,}~~\Rightarrow~~\lim_{t\rightarrow \infty}\dfrac{\mathcal{M}_{t}}{t}=0,~~\mbox{a.s.}
	\end{eqnarray*}
	
	\end{lemma}
\begin{lemma}\cite{hep}\label{lem}
	Let $\mathcal{F}\in {\cal{C}}([0,\infty)\times \Omega,(0,\infty))$ and $\mathcal{H}\in {\cal{C}}([0,\infty)\times \Omega,\mathbb{R}).$ If there exist a positive constants $a_1$, $a_2$ and $T,$ such that
	$\log \mathcal{F}(t)\geq  a_1t-a_2\displaystyle \int_0^t \mathcal{F}(x)\, \mathrm{{}}d x+\mathcal{H}(t)$ a.s, for all $t\geqslant T$ and
	$\lim\limits_{t \rightarrow +\infty}\dfrac{\mathcal{H}(t)}{t}=0$ a.s, then  $\liminf\limits_{t \rightarrow +\infty}\langle \mathcal{F}(t) \rangle \geq \dfrac{a_1}{a_2}~~a.s.$
\end{lemma}

\begin{theorem}\label{oxo}
	Let $(S(t),A(t),I(t),R(t))$  be the solution of system \eqref{sys1} that starts from a given value $(S(0),A(0),I(0),R(0))\in\mathbb{R}_+^4$.
	If the following conditions  $\mu>\dfrac{\sigma_{1}^2\vee\sigma_{2}^2\vee\sigma_{3}^2\vee\sigma_4^2}{2}$ and $\mathcal{T}_s<1$ are satisfied, then the disease goes out of the population with chance one. That is to say\\
$$\lim\limits_{t\rightarrow \infty} A(t)=\lim\limits_{t\rightarrow \infty}\langle I(t)\rangle=0$$
Furthermore,
	\begin{eqnarray*}
\lim\limits_{t\rightarrow \infty}\langle S(t)\rangle&=&\dfrac{b(\mu(1-\nu)+\xi)}{\mu(\mu+\xi+\rho)},\\
\lim\limits_{t\rightarrow \infty}\langle R(t)\rangle&=&\dfrac{b(\mu \nu+\rho)}{\mu(\mu+\xi+\rho)}.
	\end{eqnarray*}
	
\end{theorem}

\begin{proof}
We obtain the following computations by direct integration of the proposed system \eqref{sys1}.
\begin{eqnarray}\label{ka}
\dfrac{S(t)-S(0)}{t}&=&	b(1-\nu)-\dfrac{\beta_A\langle S\rangle\langle A\rangle}{\langle N\rangle}-\dfrac{\beta_I\langle S\rangle\langle I\rangle}{\langle N\rangle}-(\rho+\mu)\langle S\rangle+\xi\langle R\rangle+\dfrac{\sigma_1}{t} \displaystyle\int_{0}^{t}S(s)dB_1(s) \nonumber\\&&+\dfrac{1}{t}\displaystyle\int_0^t \int_{\mathbb{Z}}q_1(z)S(s-)\,\mathrm{{\tilde{N}}}(ds,dz),\nonumber\\
\dfrac{A(t)-A(0)}{t}&=&	\dfrac{\beta_A\langle S\rangle\langle A\rangle}{\langle N\rangle}+\dfrac{\beta_I\langle S\rangle\langle I\rangle}{\langle N\rangle}-(\gamma_A+\sigma+\mu)\langle A\rangle+\dfrac{\sigma_2}{t} \displaystyle\int_{0}^{t}A(s)dB_2(s)+\dfrac{1}{t}\displaystyle\int_0^t \int_{\mathbb{Z}}q_2(z)A(s-)\,\mathrm{{\tilde{N}}}(ds,dz),\nonumber\\
\dfrac{I(t)-I(0)}{t}&=&\sigma\langle A\rangle-(\gamma_I+\eta+\alpha+\mu)\langle I\rangle+\dfrac{\sigma_3}{t} \displaystyle\int_{0}^{t}I(s)dB_3(s)+\dfrac{1}{t}\displaystyle\int_0^t \int_{\mathbb{Z}}q_3(z)I(s-)\,\mathrm{{\tilde{N}}}(ds,dz),\nonumber\\\nonumber
\dfrac{R(t)-R(0)}{t}&=&b\nu+\gamma_A\langle A\rangle+(\gamma_I+\eta)\langle I\rangle+\rho\langle S\rangle-(\xi+\mu)\langle R\rangle+\dfrac{\sigma_4}{t} \displaystyle\int_{0}^{t}R(s)dB_4(s)\\&&+\dfrac{1}{t}\displaystyle\int_0^t \int_{\mathbb{Z}}q_4(z)R(s-)\,\mathrm{{\tilde{N}}}(ds,dz),
\end{eqnarray}	
which shows that 
\begin{equation}\label{pop}
\langle S\rangle=\dfrac{b}{\mu}-(\langle A\rangle+\langle R\rangle)-\dfrac{\alpha }{\mu}\langle I\rangle+\varUpsilon(t)
\end{equation}
where
\begin{eqnarray*}
\varUpsilon(t)&=&\dfrac{1}{\mu}\bigg\{\dfrac{\sigma_1}{t} \displaystyle\int_{0}^{t}S(s)dB_1(s)+\dfrac{1}{t}\displaystyle\int_0^t \int_{\mathbb{Z}}q_1(z)S(s-)\,\mathrm{{\tilde{N}}}(ds,dz)+\dfrac{\sigma_2}{t} \displaystyle\int_{0}^{t}A(s)dB_2(s)+\dfrac{1}{t}\displaystyle\int_0^t \int_{\mathbb{Z}}q_2(z)A(s-)\,\mathrm{{\tilde{N}}}(ds,dz)\\&&+\dfrac{\sigma_3}{t} \displaystyle\int_{0}^{t}I(s)dB_3(s)+\dfrac{1}{t}\displaystyle\int_0^t \int_{\mathbb{Z}}q_3(z)I(s-)\,\mathrm{{\tilde{N}}}(ds,dz)+\dfrac{\sigma_4}{t} \displaystyle\int_{0}^{t}R(s)dB_4(s)+\dfrac{1}{t}\displaystyle\int_0^t \int_{\mathbb{Z}}q_4(z)R(s-)\,\mathrm{{\tilde{N}}}(ds,dz)\\&&
-\bigg[\dfrac{S(t)-S(0)}{t}+\dfrac{A(t)-A(0)}{t}+\dfrac{I(t)-I(0)}{t}+\dfrac{R(t)-R(0)}{t}\bigg]\bigg\}.
\end{eqnarray*}
Applying Itô formula on the second equation of system \eqref{sys1}, we get
\begin{eqnarray}\label{zo}
d\log A(t)&=& \bigg[\beta_A\dfrac{S}{N}+\beta_I\dfrac{SI}{AN}-(\gamma_A+\sigma+\mu)-\dfrac{\sigma_2^2}{2}-\displaystyle \int_{\mathbb{Z}}[q_2(z)-\log(1+q_2(z))]\, \mathrm{}\nu (dz)\bigg] dt+\sigma_2dB_2(t)\nonumber\\&&+\displaystyle \int_{\mathbb{Z}}\log(1+q_2)\,\mathrm{{\tilde{N}}}(ds,dz),
\end{eqnarray}	
Integrating \eqref{zo} from $0$ to $t$, and then dividing by $t$ on both sides, we obtain
\begin{eqnarray}\label{op}	
\dfrac{\log A(t)-\log A(0)}{t}&=&\beta_A\dfrac{\langle S\rangle}{\langle N\rangle}+\beta_I\dfrac{\langle S\rangle \langle I \rangle}{\langle A\rangle \langle N\rangle}-(\gamma_A+\sigma+\mu)-\dfrac{\sigma_2^2}{2}+\dfrac{\sigma_2}{t}\displaystyle\int_{0}^{t}dB_2(s)\nonumber\\&&\nonumber-\displaystyle\int_{\mathbb{Z}}[q_2(z)-\log(1+q_2(z)]\,\mathrm{\nu}(dz)+\dfrac{1}{t}\displaystyle\int_0^t\int_{\mathbb{Z}}\log(1+q_2(z)\,\mathrm{\tilde{N}}(ds,dz)\\\nonumber
&\leq&(\beta_A+\beta_I)\langle S\rangle-(\gamma_A+\sigma+\mu)-\dfrac{\sigma_2^2}{2}+\dfrac{\sigma_2}{t}\displaystyle\int_{0}^{t}dB_2(s)-\displaystyle\int_{\mathbb{Z}}[q_2(z)-\log(1+q_2(z)]\,\mathrm{\nu}(dz)\\&&+\dfrac{1}{t}\displaystyle\int_0^t\int_{\mathbb{Z}}\log(1+q_2(z)\,\mathrm{\tilde{N}}(ds,dz).
\end{eqnarray}	
Combining \eqref{pop} with \eqref{op} yields
\begin{eqnarray*}
\dfrac{\log A(t)-\log A(0)}{t}&\leq & (\beta_A+\beta_I)(\dfrac{b}{\mu}-(\langle A\rangle+\langle R\rangle)-\dfrac{\alpha }{\mu}\langle I\rangle+\varUpsilon)-(\gamma_A+\sigma+\mu)-\dfrac{\sigma_2^2}{2}+\dfrac{\sigma_2}{t}\displaystyle\int_{0}^{t}dB_2(s)\\&&-\displaystyle\int_{\mathbb{Z}}[q_2(z)-\log(1+q_2(z)]\,\mathrm{\nu}(dz)+\dfrac{1}{t}\displaystyle\int_0^t\int_{\mathbb{Z}}\log(1+q_2(z)\,\mathrm{\tilde{N}}(ds,dz)\\
&\leq & \dfrac{b}{\mu}(\beta_A+\beta_I)-\bigg[\gamma_A+\sigma+\mu+\dfrac{\sigma_2^2}{2}+\displaystyle\int_{\mathbb{Z}}[q_2(z)-\log(1+q_2(z)]\,\mathrm{\nu}(dz)\bigg]\\&&+\dfrac{\sigma_2}{t}\displaystyle\int_{0}^{t}dB_2(s)+\dfrac{1}{t}\displaystyle\int_0^t\int_{\mathbb{Z}}\log(1+q_2(z)\,\mathrm{\tilde{N}}(ds,dz)+ (\beta_A+\beta_I)\varUpsilon(t).
\end{eqnarray*}	
Let $\mathcal{G}(t)=\dfrac{\sigma_2}{2}\displaystyle\int_{0}^{t}dB_2(s)+\displaystyle\int_0^t\int_{\mathbb{Z}}\log(1+q_2(z)\,\mathrm{\tilde{N}}(ds,dz)$. \\ which is known as the locally continuous martingale with a finite quadratic variation, and the similar way we can affirm that $\varUpsilon(t)$ will be also so. \\
Then, in virtue to the Lemma \ref{strg}, we obtain
\begin{eqnarray*}
\limsup\limits_{t\rightarrow\infty}\dfrac{\mathcal{G}(t)}{t}=0~~\mbox{and}~~\limsup\limits_{t\rightarrow\infty}\dfrac{\varUpsilon(t)}{t}=0.	
\end{eqnarray*}
If $\mathcal{T}_s<1$, we  get
\begin{eqnarray*}
\limsup\limits_{t\rightarrow\infty}\dfrac{\log A(t)}{t}&\leq& \bigg[\gamma_A+\sigma+\mu+\dfrac{\sigma_2^2}{2}+\displaystyle\int_{\mathbb{Z}}[q_2(z)-\log(1+q_2(z)]\,\mathrm{\nu}(dz)\bigg]\bigg(\mathcal{T}_s-1\bigg)<0~~\mbox{a.s.}
\end{eqnarray*}
which implies that 
\begin{eqnarray}\label{ui}
\lim_{t\rightarrow \infty}A(t)=0~~\mbox{a.s.}
\end{eqnarray}
Furthermore, we will use Eq. \eqref{ui} to solve the 3rd classe of the model \eqref{sys1}. Combined with the restriction of integration with a limit, that is,
from $0$ to $t$, and then divide it by $t$, we obtain
\begin{eqnarray*}
\dfrac{I(t)-I(0)}{t}&=&\sigma\langle A\rangle-(\gamma_I+\eta+\alpha+\mu)\langle I\rangle+\dfrac{\sigma_3}{t} \displaystyle\int_{0}^{t}I(s)dB_3(s)+\dfrac{1}{t}\displaystyle\int_0^t \int_{\mathbb{Z}}q_3(z)I(s-)\,\mathrm{{\tilde{N}}}(ds,dz).
\end{eqnarray*}
Hence
\begin{eqnarray*}
\langle I(t)\rangle&=&\dfrac{1}{\gamma_I+\eta+\alpha+\mu}\bigg[\sigma\langle A(t)\rangle+\dfrac{I(t)-I(0)}{t}+\dfrac{\sigma_3}{t}\displaystyle\int_{0}^{t}I(s)dB_3(s)
+\dfrac{1}{t}\displaystyle\int_0^t\int_{\mathbb{Z}}q_3(z)I(t-)\,\mathrm{\tilde{N}}(ds,dz)\bigg].
\end{eqnarray*}
which implies that $\lim_{t\rightarrow \infty}\langle I(t)\rangle=0$ a.s.
\\
Let us calculate the following equation by using Eq. \refeq{ka},
\begin{eqnarray*}
\dfrac{S(t)-S(0)+A(t)-A(0)+R(t)-R(0)}{t}&=& b-\mu\langle S(t)\rangle-(\sigma+\mu) \langle A(t)\rangle+ (\gamma_I+\eta)\langle I(t)\rangle -\mu\langle R(t)\rangle\\&&+\dfrac{\sigma_1}{t} \displaystyle\int_{0}^{t}S(s)dB_1(s) +\dfrac{1}{t}\displaystyle\int_0^t\int_{\mathbb{Z}}q_1(z)S(s-)\,\mathrm{{\tilde{N}}}(ds,dz)\\&&+\dfrac{\sigma_2}{t} \displaystyle\int_{0}^{t}A(s)dB_2(s)+\dfrac{1}{t}\displaystyle\int_0^t \int_{\mathbb{Z}}q_2(z)A(s-)\,\mathrm{{\tilde{N}}}(ds,dz)\\&&+\dfrac{\sigma_4}{t} \displaystyle\int_{0}^{t}R(s)dB_4(s)+\dfrac{1}{t}\displaystyle\int_0^t \int_{\mathbb{Z}}q_4(z)R(s-)\,\mathrm{{\tilde{N}}}(ds,dz),
\end{eqnarray*}
we get $\langle S(t)+R(t)\rangle=\dfrac{b}{\mu}+\varXi(t)$.\\
The operator $\varXi(t)$ is defined as follows:
\begin{eqnarray*}
\varXi(t)&=&-\dfrac{1}{\mu}\bigg[\dfrac{S(t)-S(0)}{t}-\dfrac{\sigma_1}{t} \displaystyle\int_{0}^{t}S(s)dB_1(s)-\dfrac{1}{t}\displaystyle\int_0^t\int_{\mathbb{Z}}q_1(z)S(s-)\,\mathrm{{\tilde{N}}}(ds,dz)+\dfrac{A(t)-A(0)}{t}\\&&-\dfrac{\sigma_2}{t} \displaystyle\int_{0}^{t}A(s)dB_2(s)-\dfrac{1}{t}\displaystyle\int_0^t \int_{\mathbb{Z}}q_2(z)A(s-)\,\mathrm{{\tilde{N}}}(ds,dz)+\dfrac{R(t)-R(0)}{t}-\dfrac{\sigma_4}{t} \displaystyle\int_{0}^{t}R(s)dB_4(s)\\&&-\dfrac{1}{t}\displaystyle\int_0^t \int_{\mathbb{Z}}q_4(z)R(s-)\,\mathrm{{\tilde{N}}}(ds,dz)\bigg].
\end{eqnarray*}
Clearly, $\varXi(t)$ goes to zero is the same as $t$ goes to $\infty$. So that we can\\
$\langle S(t)\rangle+\langle R(t)\rangle=\dfrac{b}{\mu}$\\
Similarly, we get by using the last equation of the system \refeq{ka},
\begin{eqnarray}\label{exx}
\rho\langle S(t)\rangle+\gamma_A\langle A(t)\rangle+(\gamma_I+\eta)\langle I(t)\rangle-(\xi+\mu)\langle R(t)\rangle&=&-b\nu-\dfrac{R(t)-R(0)}{t}-\dfrac{\sigma_4}{t} \displaystyle\int_{0}^{t}R(s)dB_4(s)\nonumber\\&&-\dfrac{1}{t}\displaystyle\int_0^t \int_{\mathbb{Z}}q_4(z)R(s-)\,\mathrm{{\tilde{N}}}(ds,dz)
\end{eqnarray}
By taking t goes to $\infty$, Eq. \ref{exx} can be as

\begin{eqnarray}\label{exy}
	\rho\langle S(t)\rangle\rangle-(\xi+\mu)\langle R(t)\rangle=-b\nu
\end{eqnarray}

Thus, we obtain the following result from Eqs. \ref{exx} to \ref{exy}

\begin{eqnarray}
\lim\limits_{t\rightarrow \infty}\langle S(t)\rangle=\dfrac{b(\mu(1-\nu)+\xi)}{\mu(\mu+\xi+\rho)},
\end{eqnarray}
and
\begin{eqnarray}
	\lim\limits_{t\rightarrow \infty}\langle R(t)\rangle=\dfrac{b(\mu \nu+\rho)}{\mu(\mu+\xi+\rho)}.
\end{eqnarray}
This completes the proof. 

\end{proof}

\subsection{Persistence in mean}
After having studied the extinction of the disease, we turn now to explore its persistence in the mean, but before stating the main result, we will firstly define the persistency in the average.
\begin{definition}
	Persistence in the mean. For system \eqref{sys1}, the infectious individuals $A(t)$ and $I(t)$ are said
	to be strongly persistent in the mean, or just persistent in the mean, if $\liminf\limits_{t\rightarrow \infty}\langle A(t)+I(t)\rangle  >0$
\end{definition}

\begin{theorem}\label{opo}
Let $(S(t),A(t),I(t),R(t))$  be the solution of system \eqref{sys1} that starts from an initial value $(S(0),A(0),I(0),R(0))\in\mathbb{R}_+^4$. If $\mathcal{T}_s^*>1$, then the  disease presented by equation \eqref{sys1} will persist in the mean almost surely.	
\end{theorem}
\begin{proof}
Let us consider the function $\mathcal{W}$ the function defined by
\begin{eqnarray*}
\mathcal{W}:\mathbb{R}^2_+&\longrightarrow &\mathbb{R}\\
(S,I)&\longrightarrow &-\tau_{1} \log S-\tau_{2} \log A
\end{eqnarray*} 
By the application of Itô's formula, we get
	\begin{eqnarray*}
		d \mathcal{W}=\mathcal{L} \mathcal{W}dt-\tau_{1} \sigma_{1} d B_{1}(t)-\tau_{2} \sigma_{2} d B_{2}(t)-\tau_1\displaystyle \int_{\mathbb{Z}}\log (1+q_1(z))\, \mathrm{{\tilde{N}}}(ds,dz)-\tau_2\displaystyle \int_{\mathbb{Z}}\log (1+q_2(z))\, \mathrm{{\tilde{N}}}(ds,dz),
	\end{eqnarray*}
	where
	\begin{eqnarray*}
		\mathcal{L} \mathcal{W}&=&-\tau_{1}\mathcal{L} ( \log S)-\tau_{2}\mathcal{L} ( \log A)\\
		&=&\tau_{1}\dfrac{\beta_A }{N}+\tau_{1}\dfrac{\beta_I I}{N}-\dfrac{\tau_{1}b(1-\nu)}{S}-\tau_{1}\dfrac{\xi R}{S}+\tau_{1}(\rho+\mu)+\tau_1\bigg\{\dfrac{\sigma_1^2}{2}+\displaystyle \int_{\mathbb{Z}}[q_1(z)-\log(1+
		q_1(z))]\, \mathrm{}\nu (dz) \bigg\}\\&&
		-\tau_{2}\dfrac{\beta_A S}{N}-\tau_{2}\dfrac{\beta_I S I}{N}+\tau_{2}(\gamma_A+\sigma+\mu)+\tau_2\bigg\{\dfrac{\sigma_2^2}{2}+\displaystyle \int_{\mathbb{Z}}[q_2(z)-\log(1+
		q_2(z))]\, \mathrm{}\nu (dz) \bigg\}\\
		&\leq & (-\dfrac{\tau_{1}b(1-\nu)}{S}-\tau_2\beta_A S)+\tau_1\bigg[\rho+\mu+\dfrac{\sigma_1^2}{2}+\displaystyle \int_{\mathbb{Z}}[q_1(z)-\log(1+
		q_1(z))]\, \mathrm{}\nu (dz)\bigg]\\&&
		+\tau_{2}\bigg[\gamma_A+\sigma+\mu +\dfrac{\sigma_2^2}{2}+\displaystyle \int_{\mathbb{Z}}[q_2(z)-\log(1+
		q_2(z))]\, \mathrm{}\nu (dz)\bigg]+\tau_1\beta_A+\tau_1\beta_I I\\
		&\leq & -2\sqrt{\tau_1\tau_2b(1-\nu)\beta_A}+\tau_1\bigg[\rho+\mu+\dfrac{\sigma_1^2}{2}+\displaystyle \int_{\mathbb{Z}}[q_1(z)-\log(1+
		q_1(z))]\, \mathrm{}\nu (dz)\bigg]\\&&
		+\tau_{2}\bigg[\gamma_A+\sigma+\mu +\dfrac{\sigma_2^2}{2}+\displaystyle \int_{\mathbb{Z}}[q_2(z)-\log(1+
		q_2(z))]\, \mathrm{}\nu (dz)\bigg]+\tau_1\beta_A+\tau_1\beta_I I.
\\
	&\leq & -2\sqrt{\dfrac{b(1-\nu)^2\beta_Ab(1-\nu)}{\kappa_1\kappa_2}}+2b(1-\nu)+\tau_1\beta_AA+\tau_1\beta_II\\
		&=&-2b(1-\nu)\bigg[\sqrt{\dfrac{\beta_Ab(1-\nu)}{\kappa_1\kappa_2}}-1\bigg]+\tau_1\beta_AA+\tau_1\beta_II\\
		&=&-2b(1-\nu)\bigg[\sqrt{\mathcal{T}^*_s}-1\bigg]+\tau_1\beta_AA+\tau_1\beta_II.
	\end{eqnarray*}
Therefore
	\begin{eqnarray}\label{ouk}
		\dfrac{\mathcal{W}(S(t),A(t))-\mathcal{W}(S(0),A(0))}{t}&\leq&-2b(1-\nu)(\sqrt{\mathcal{T}^*_s}-1)+\tau_1\beta_A\langle A(t)\rangle-\dfrac{\tau_1\sigma_1 B_1(t)}{t}-\dfrac{\tau_2\sigma_2 B_2(t)}{t}\nonumber\\&&\nonumber
		-\dfrac{\tau_1\displaystyle\int_0^t \int_{\mathbb{Z}}\log (1+q_1(z))\, \mathrm{{\tilde{N}}}(ds,dz)}{t}-\dfrac{\tau_2\displaystyle\int_0^t \int_{\mathbb{Z}}\log (1+q_2(z))\, \mathrm{{\tilde{N}}}(ds,dz)}{t}\\
		&\leq&-2b(1-\nu)(\sqrt{\mathcal{T}^*_{s}}-1)+\tau_1\beta_A\langle A(t)\rangle+\tau_1\beta_I\langle I(t)\rangle+\psi(t),
	\end{eqnarray}
	where $$\psi(t)=-\dfrac{\tau_1\sigma_1B_1(t)}{t}-\dfrac{\tau_2\sigma_2 B_2(t)}{t}-\dfrac{\tau_1\displaystyle\int_0^t \int_{\mathbb{Z}}\log (1+q_1(z))\,\mathrm{{\tilde{N}}}(ds,dz)}{t}-\dfrac{\tau_2\displaystyle\int_0^t \int_{\mathbb{Z}}\log (1+q_2(z))\,\mathrm{{\tilde{N}}}(ds,dz)}{t}.$$
On the other hand it is clear by virtue of Lemma \ref{strg} that
\begin{eqnarray}\label{aso}
	\lim\limits_{t\rightarrow \infty}\psi(t)=0,	
\end{eqnarray}	
From Eq. \ref{ouk}, we have
\begin{eqnarray}\label{houd}
\tau_1\beta_A\langle A(t)\rangle+\tau_1\beta_I\langle I(t)\rangle\geq 2b(1-\nu)(\sqrt{\mathcal{T}^*_s}-1)-\psi(t)+\dfrac{\mathcal{W}(S(t),A(t))-\mathcal{W}(S(0),A(0))}{t}	
\end{eqnarray}
By Lemma \ref{lem} and Eq. \ref{aso}  and Taking the inferior limit on both sides of \eqref{houd} yields
	$$\liminf\limits_{t\rightarrow\infty}(\langle A(t)\rangle+\langle I(t)\rangle)\geq \dfrac{2b(1-\nu)(\sqrt{\mathcal{T}^*_s}-1)}{\tau_1\beta} $$
	where $\beta=\max(\beta_A,\beta_I)$. \\\
So if $\mathcal{T}^*_s> 1$,  then the disease will persist in the mean as claimed, which completes the proof. 	
\end{proof}

\section{Numerical simulation}
In what follows, we apply our results to real ongoing COVID-19 pandemic data. Most of the parametric values appearing in Table \ref{Tab1} are selected from real data available in existing literature (\cite{1,nmr,nmr1,nmr2}) and the rest of them are just assumed for numerical calculations.  We employ the Euler-Maruyama scheme to seek a numerical solution to the continuous part of \eqref{sys1} and the scheme presented in \cite{num} for the Lévy jumps part. We choose $q_i = 0.198\dfrac{ z}{1+z^2},~z=0.5~~ \forall i=1,...,4$ which verifies all the previous sections conditions.
\begin{table}[h!]
	\begin{center}
		\begin{center}
			\caption{Table of parameters used in the numerical simulation}\label{Tab1}
		\end{center}
		\begin{tabular}{llll}
			Notation & Parameter description           &  Value range         \\
			\hline
			\hline
			$b$         & Influx rate of the population &   0.8        \\
			$\mu$            &  Natural death rate   & 0.1      \\ 
			$\alpha$       &  Disease-induced death rate    &    0.017 \\
			$\beta_A$          &  Transmission rate of infection from asymptomatic cases   &
			$(1.7 \times 10^{-9} ,5.2\times 10^{-3}) $         \\
			$\beta_I$          &   Transmission rate of infection from symptomatic cases   & [0,$\beta_A$]          \\
			$\gamma_A$       & Self-recovery rate  of asymptomatic individuals     &    0.15
			\\ 
			$\gamma_I$       &  Self-recovery rate of symptomatic individuals      &  0.1001 \\
			$\xi$       &  Loss of immunity rate of recovered individuals    &    
			0.42\\
			$\sigma$      & Progression rate of asymptomatic to symptomatic compartment     &   0.29\\
			$\eta$       &  Treatment rate of symptomatic individuals    &  0.53
			\\
			\hline
			\hline    
			
		\end{tabular}
	\end{center}
\end{table}
\begin{example}
For the data of Table \ref{Tab1} with the values $\nu=0.32$, $\rho=0.19$ and the quadruplets $(\sigma_{1},\sigma_{2}, \sigma_{3}, \sigma_{4}) = (0.1, 0.06, 0.08, 0.05)$, a direct computation, we derive $\mathcal{T}_s= 0.2974 < 1$  that is all less than one and $\mu=0.1>\dfrac{\sigma_{1}^2\vee\sigma_{2}^2\vee\sigma_{3}^2\vee\sigma_4^2}{2}=0.005.$ Consequently, the conditions of Theorem \ref{oxo} are satisfied, and the disease is expected to die out in the population. This result is well illustrated by the trajectories shown in Figs. \ref{fig:susceptible}–\ref{fig:recovered}, which depict the dynamics of the Susceptible (S), Asymptomatic (A), Infected (I), and Recovered (R) populations over time. The initial values used in the simulation are as follows: \textbf{Initial Susceptible Population} (\(S_0\)): \( 0.5 \), \textbf{Initial Asymptomatic Population} (\(A_0\)): \( 0.2 \), \textbf{Initial Infected Population} (\(I_0\)): \( 0.1 \), and \textbf{Initial Recovered Population} (\(R_0\)): \( 0.05 \). The time-series plots show the temporal evolution of each compartment, highlighting key trends such as the initial decline in the susceptible population due to infection spread, followed by subsequent recovery phases. A 3D surface plot provides a comprehensive view of all compartments over time, revealing intricate inter-compartment relationships and transitions. The heatmap scatter plot between the susceptible and recovered populations further illustrates the dynamic interaction, indicating the flow between these states. Finally, histograms of each compartment offer insights into the distribution and frequency of the population states, emphasizing the variability and stochastic nature of the epidemic's progression.

\end{example}

\begin{figure}[H]
    \centering
    \begin{minipage}{0.45\linewidth}
        \centering
        \includegraphics[width=\linewidth]{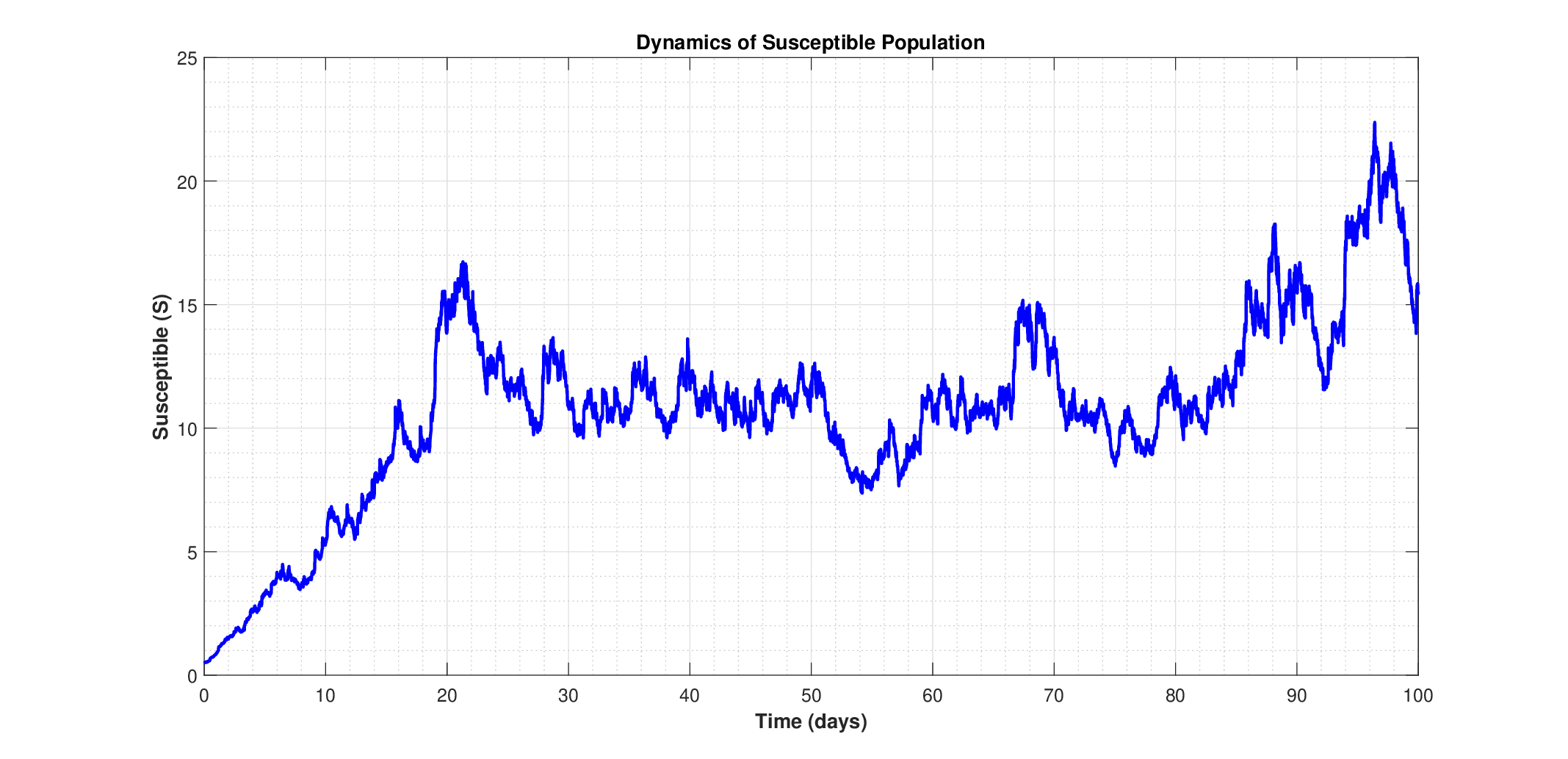}
        \caption{Dynamics of the Susceptible (S) \\ population over time.}
        \label{fig:susceptible}
    \end{minipage}\hfill
    \begin{minipage}{0.45\linewidth}
        \centering
        \includegraphics[width=\linewidth]{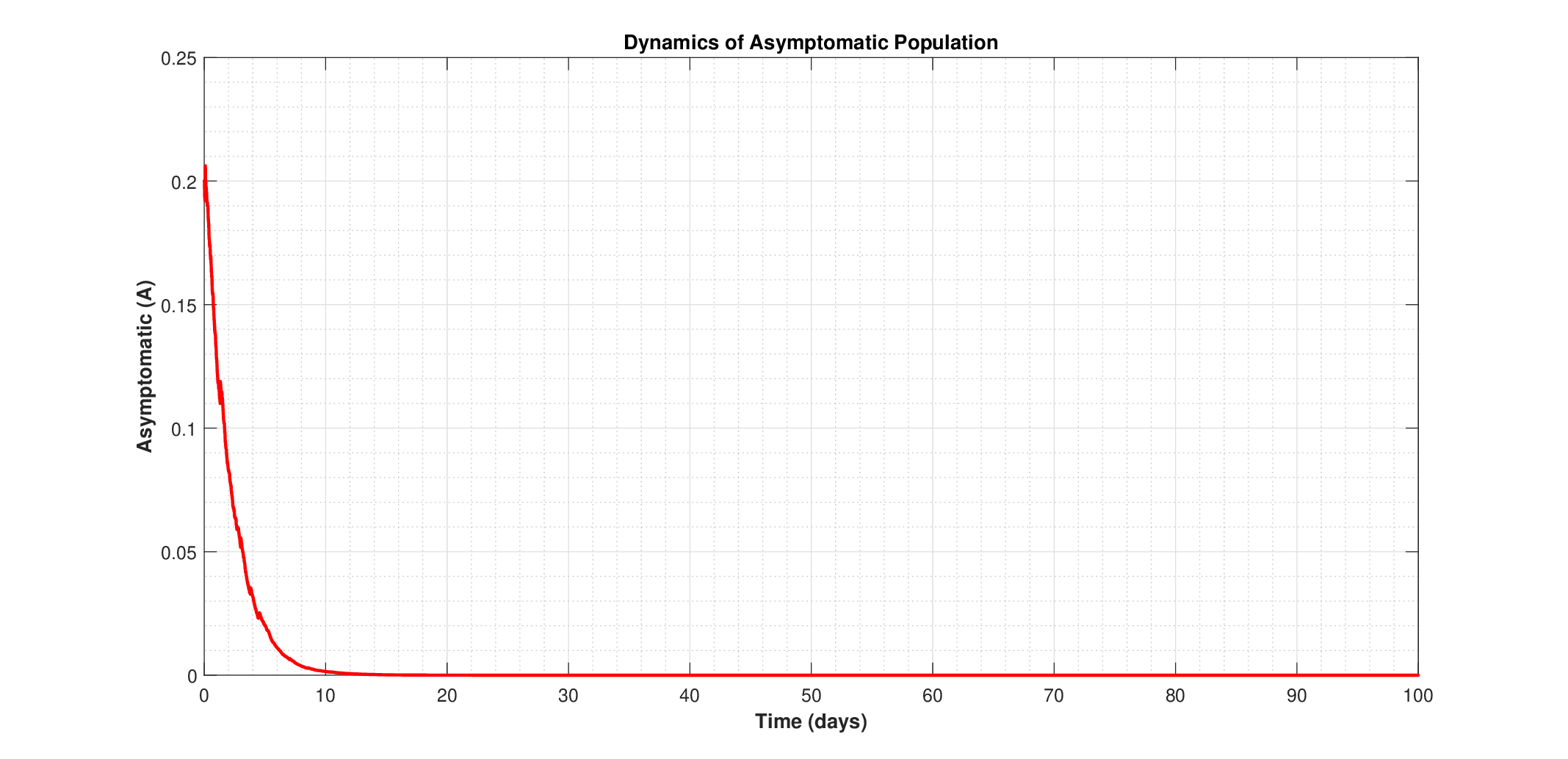}
        \caption{Dynamics of the Asymptomatic \\ (A) population over time.}
        \label{fig:asymptomatic}
    \end{minipage}

    \vspace{0.5cm} 

    \begin{minipage}{0.45\linewidth}
        \centering
        \includegraphics[width=\linewidth]{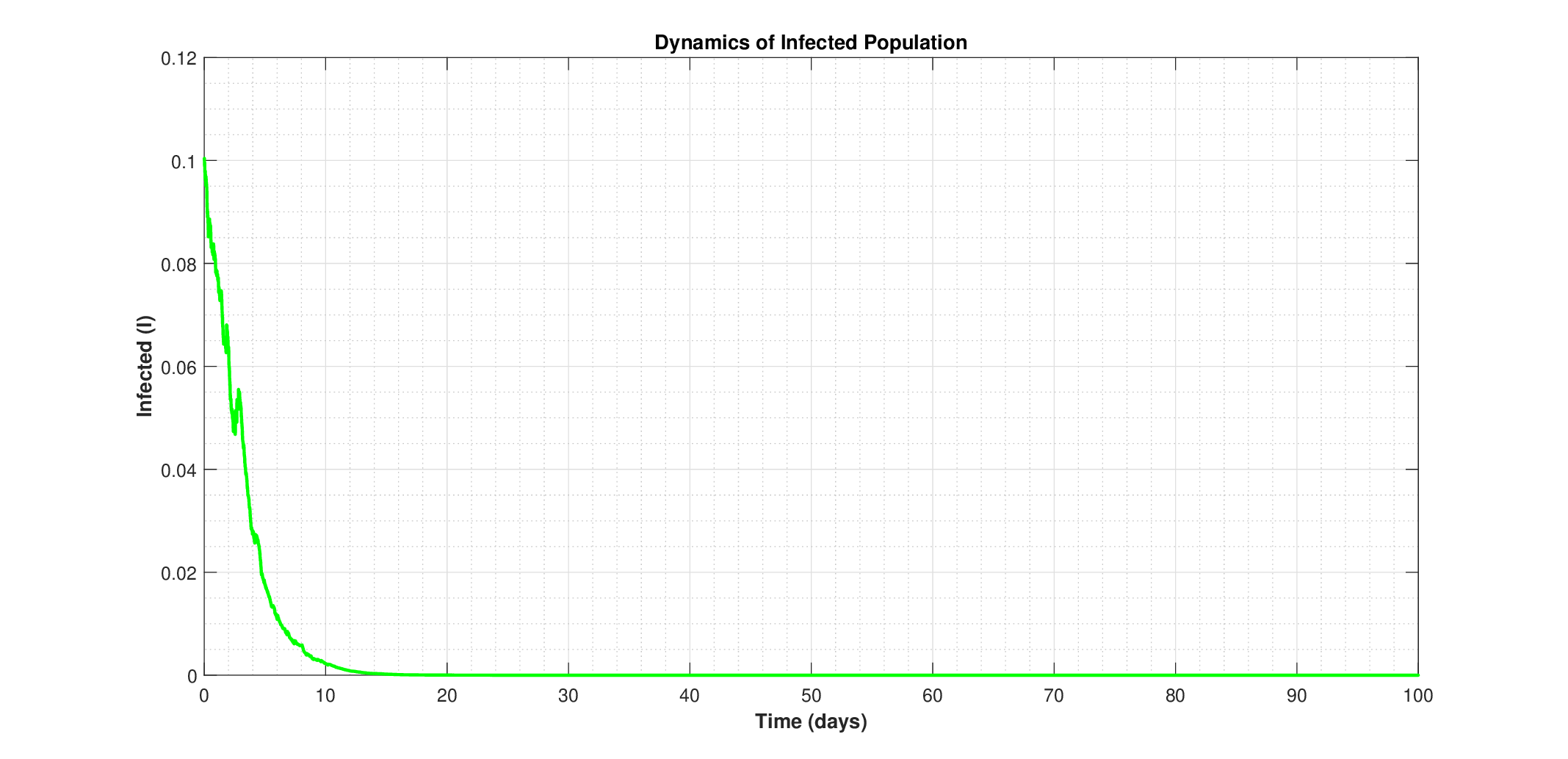}
        \caption{Dynamics of the Symptomatic (I) \\ population over time.}
        \label{fig:symptomatic}
    \end{minipage}\hfill
    \begin{minipage}{0.45\linewidth}
        \centering
        \includegraphics[width=\linewidth]{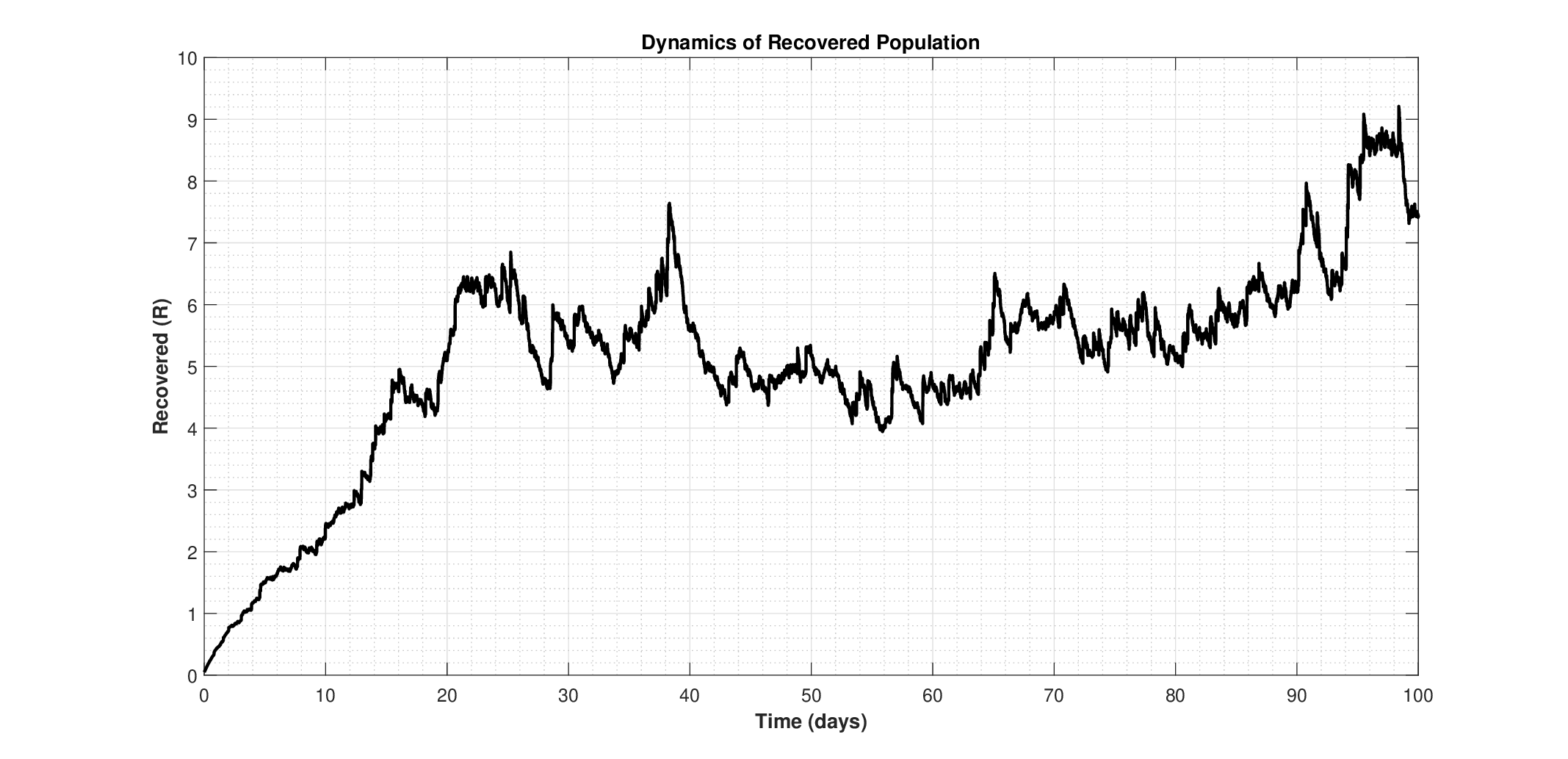}
        \caption{Dynamics of the Recovered (R) \\ population over time.}
        \label{fig:recovered}
    \end{minipage}
\end{figure}

\begin{figure}[H]
    \centering
    \includegraphics[width=\linewidth]{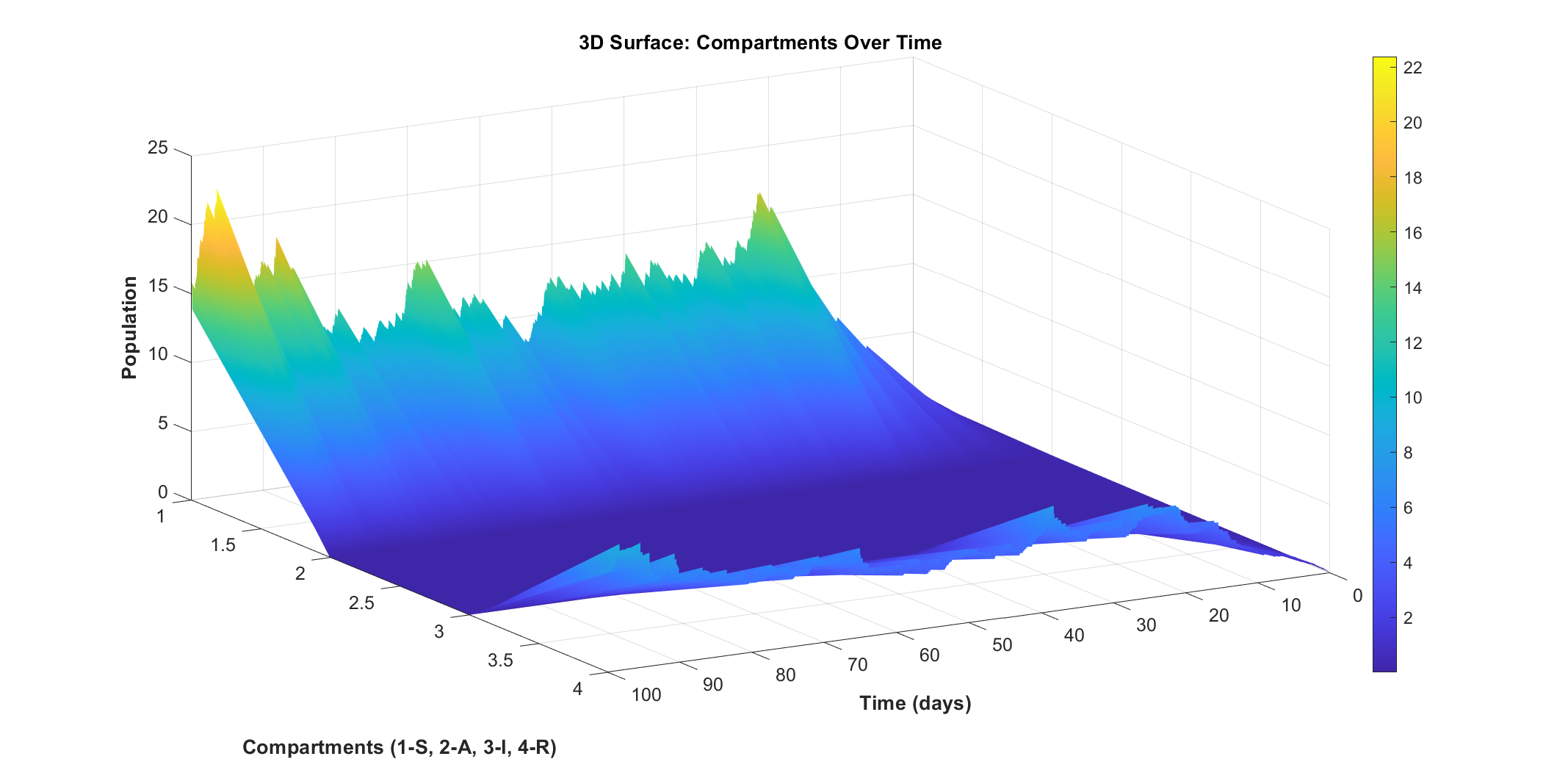}
    \caption{3D surface plot showing the evolution of all compartments (S, A, I, R) over time. This plot provides a comprehensive view of the transitions and interactions between compartments throughout the simulation period.}
    \label{fig:3d_surface}
\end{figure}

\begin{figure}[H]
    \centering
    \includegraphics[width=\linewidth]{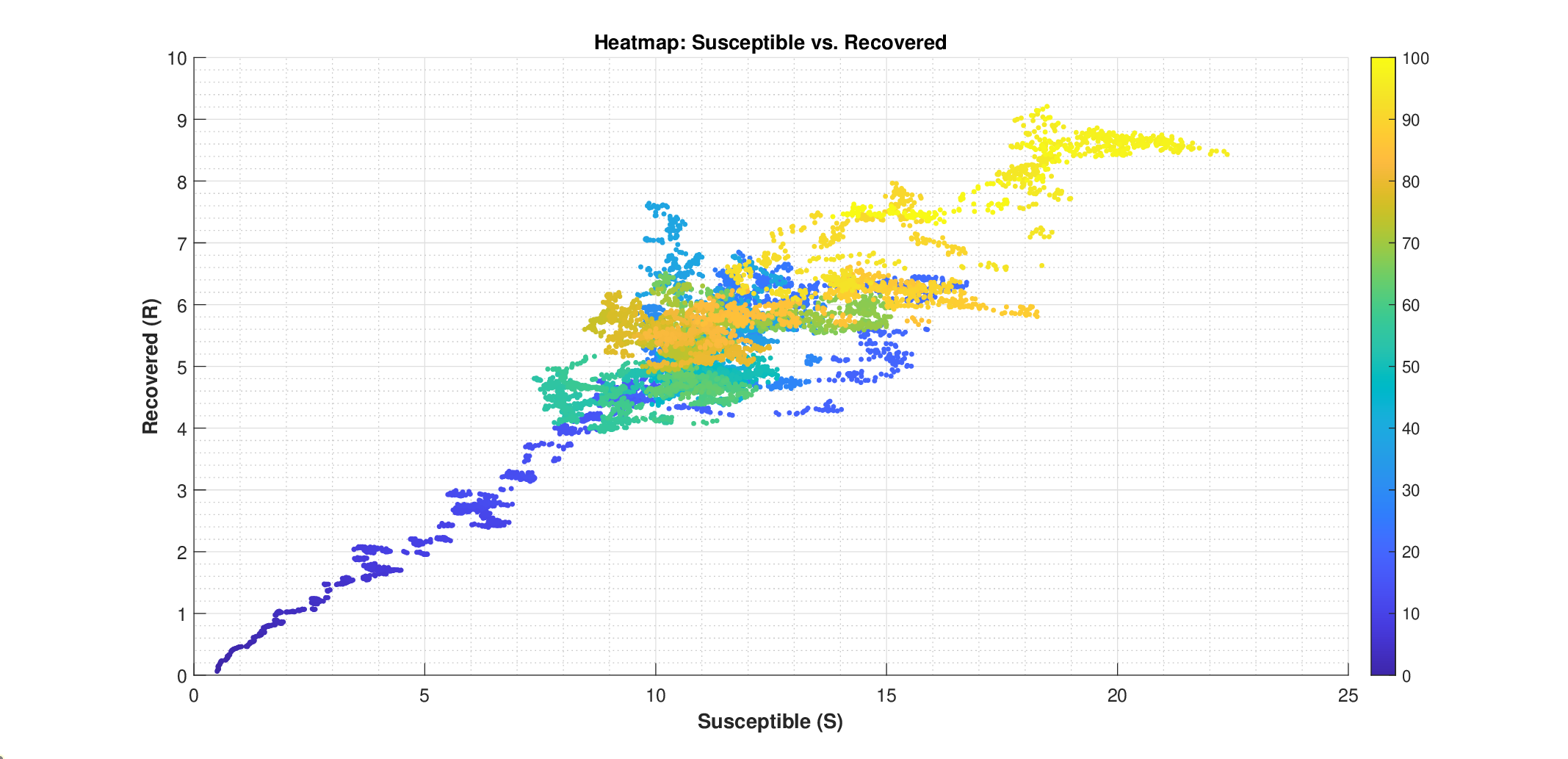}
    \caption{Heatmap of Susceptible (S) vs. Recovered (R) populations. The scatter plot, colored by time, illustrates the relationship and transitions between susceptible and recovered states, showing the progression of individuals from susceptibility to recovery.}
    \label{fig:heatmap}
\end{figure}

\begin{figure}[H]
    \centering
    \includegraphics[width=\linewidth]{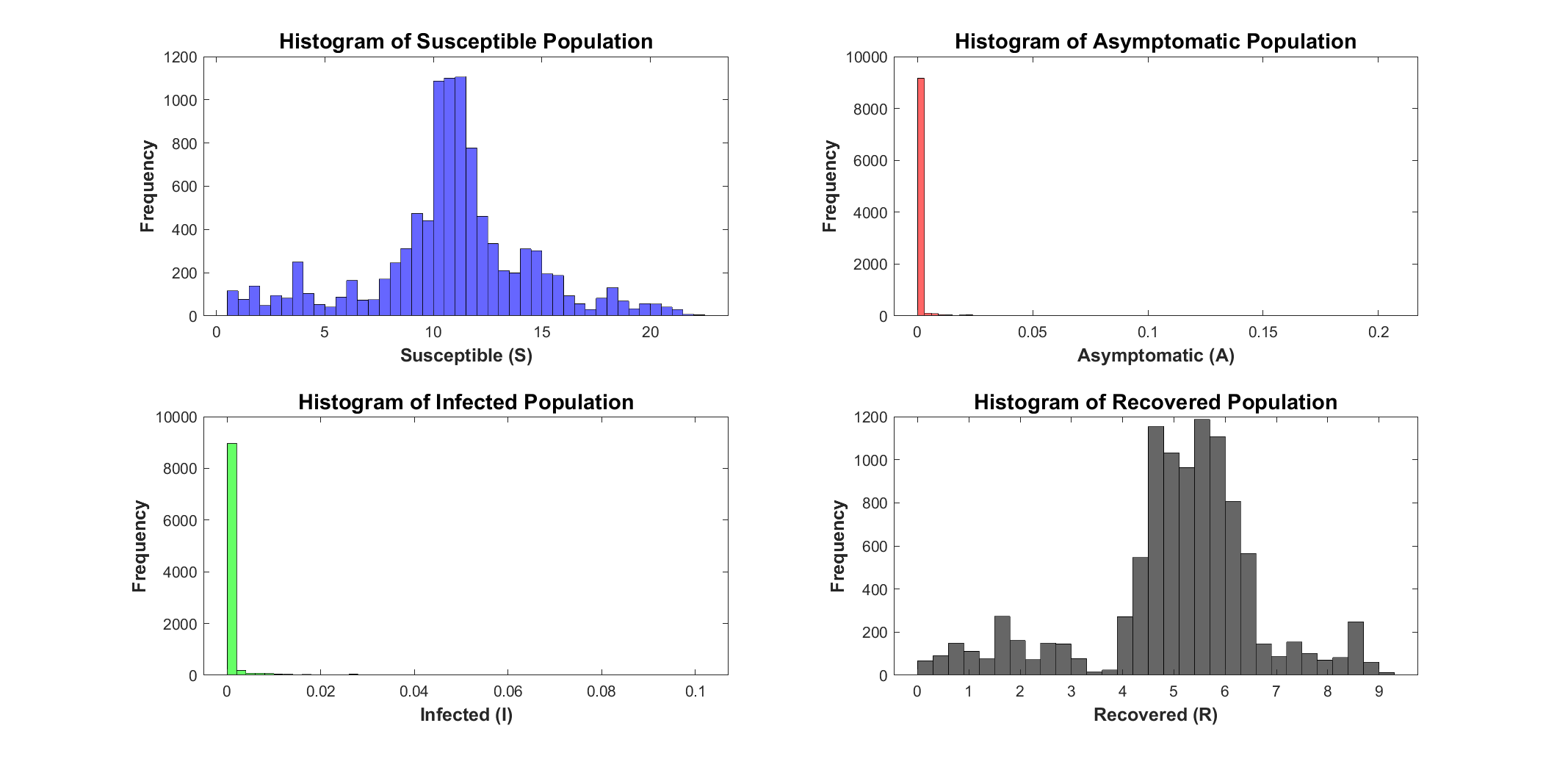}
    \caption{Histograms of the distribution of each compartment (S, A, I, R). These histograms show the frequency distribution of the populations in each compartment, highlighting the variability and stochastic behavior in the epidemic's progression.}
    \label{fig:histograms}
\end{figure}

\begin{example}
 Updating some parameter values by letting  $\nu=0.2$, $\rho=0.17$ and $(\sigma_{1},\sigma_{2}, \sigma_{3}, \sigma_{4}) = (0.15,0.1,0.1,0.2)$, the threshold quantity takes the value  $\mathcal{T}_s^*= 1.9343 >1$ satisfying the conditions of Theorem \ref{opo}. 
Then, numerical simulations in such scenarios have to show positive values of infected numbers throughout the time and this fact is clearly shown by Figs. 8-11, where the trajectories of infectives are bounded beyond zero.
\end{example}

\begin{figure}[H]
    \centering
    \begin{minipage}{0.45\linewidth}
        \centering
        \includegraphics[width=\linewidth]{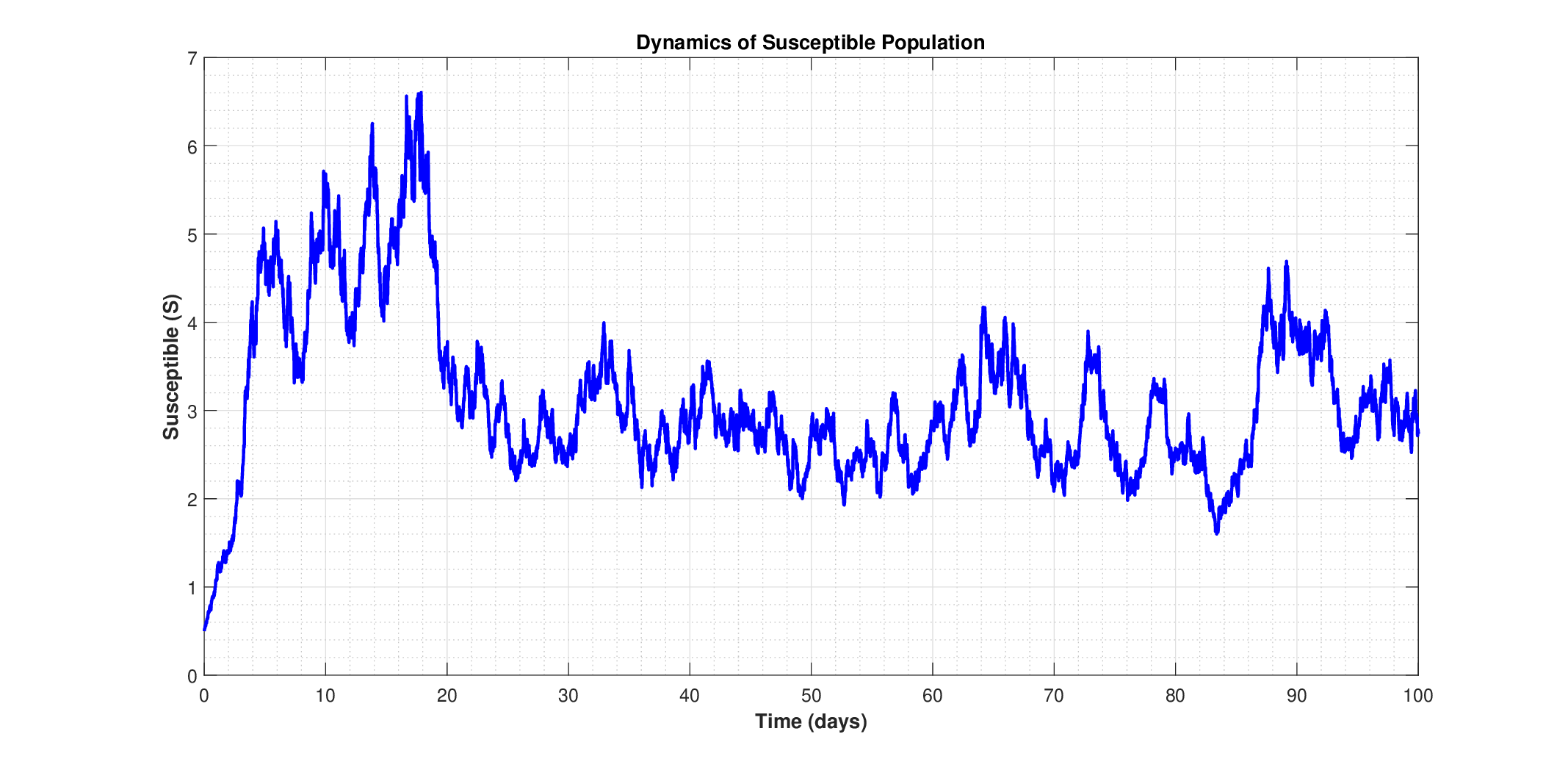}
        \caption{Dynamics of the Susceptible (S) \\ population over time.}
        \label{fig:susceptible2}
    \end{minipage}\hfill
    \begin{minipage}{0.45\linewidth}
        \centering
        \includegraphics[width=\linewidth]{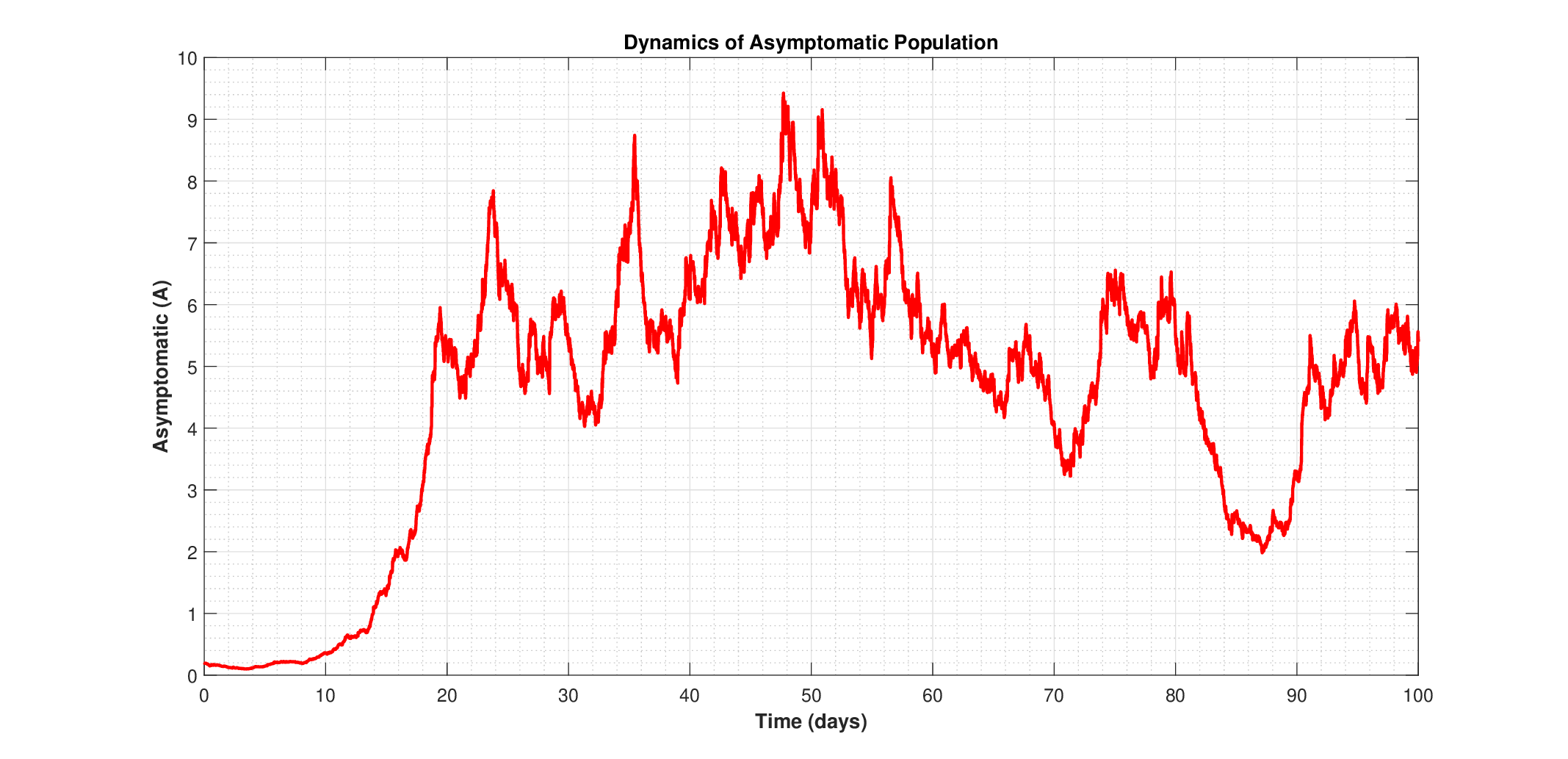}
        \caption{Dynamics of the Asymptomatic \\ (A) population over time.}
        \label{fig:asymptomatic2}
    \end{minipage}

    \vspace{0.5cm} 

    \begin{minipage}{0.45\linewidth}
        \centering
        \includegraphics[width=\linewidth]{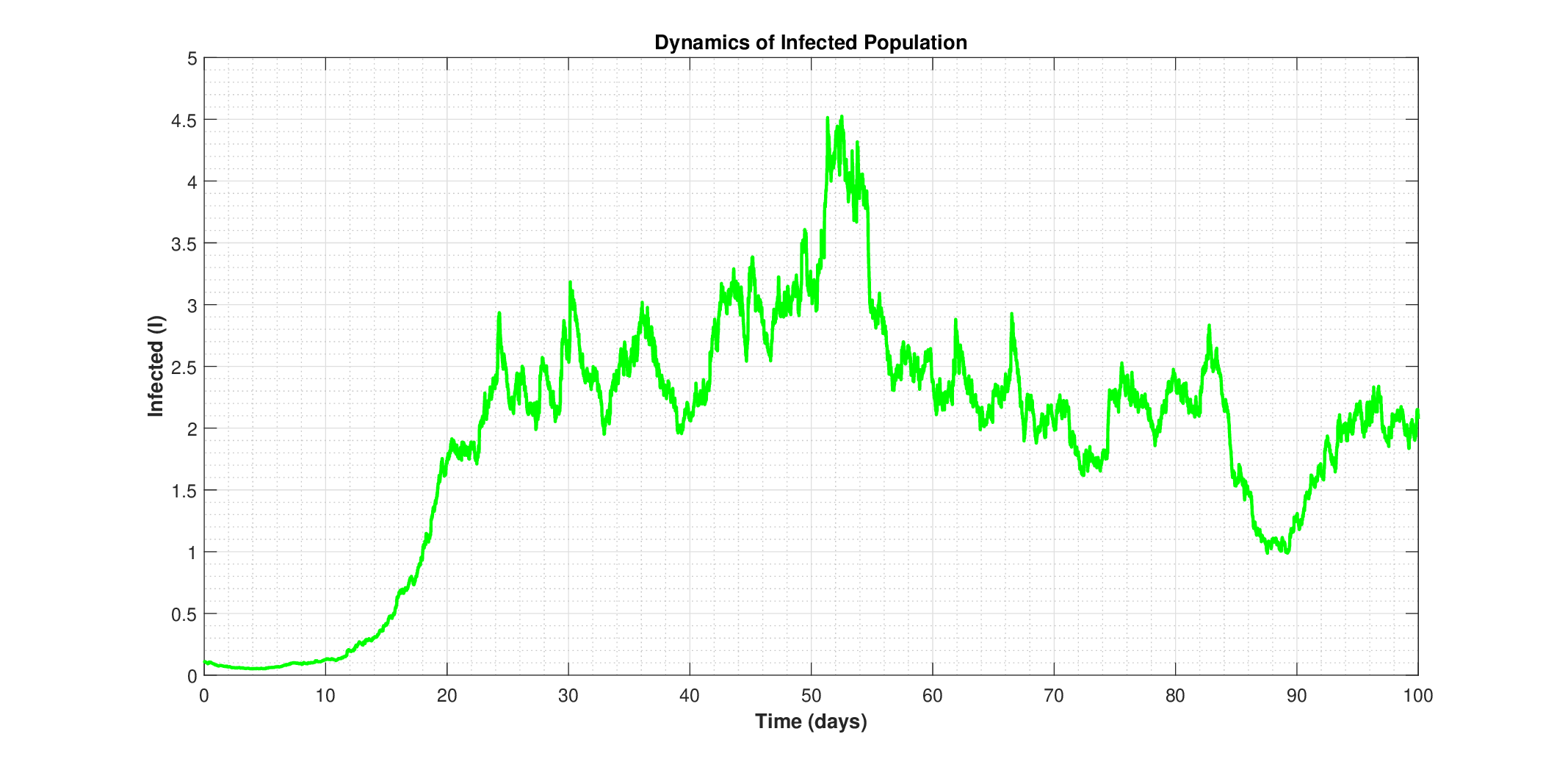}
        \caption{Dynamics of the Symptomatic (I) \\ population over time.}
        \label{fig:symptomatic2}
    \end{minipage}\hfill
    \begin{minipage}{0.45\linewidth}
        \centering
        \includegraphics[width=\linewidth]{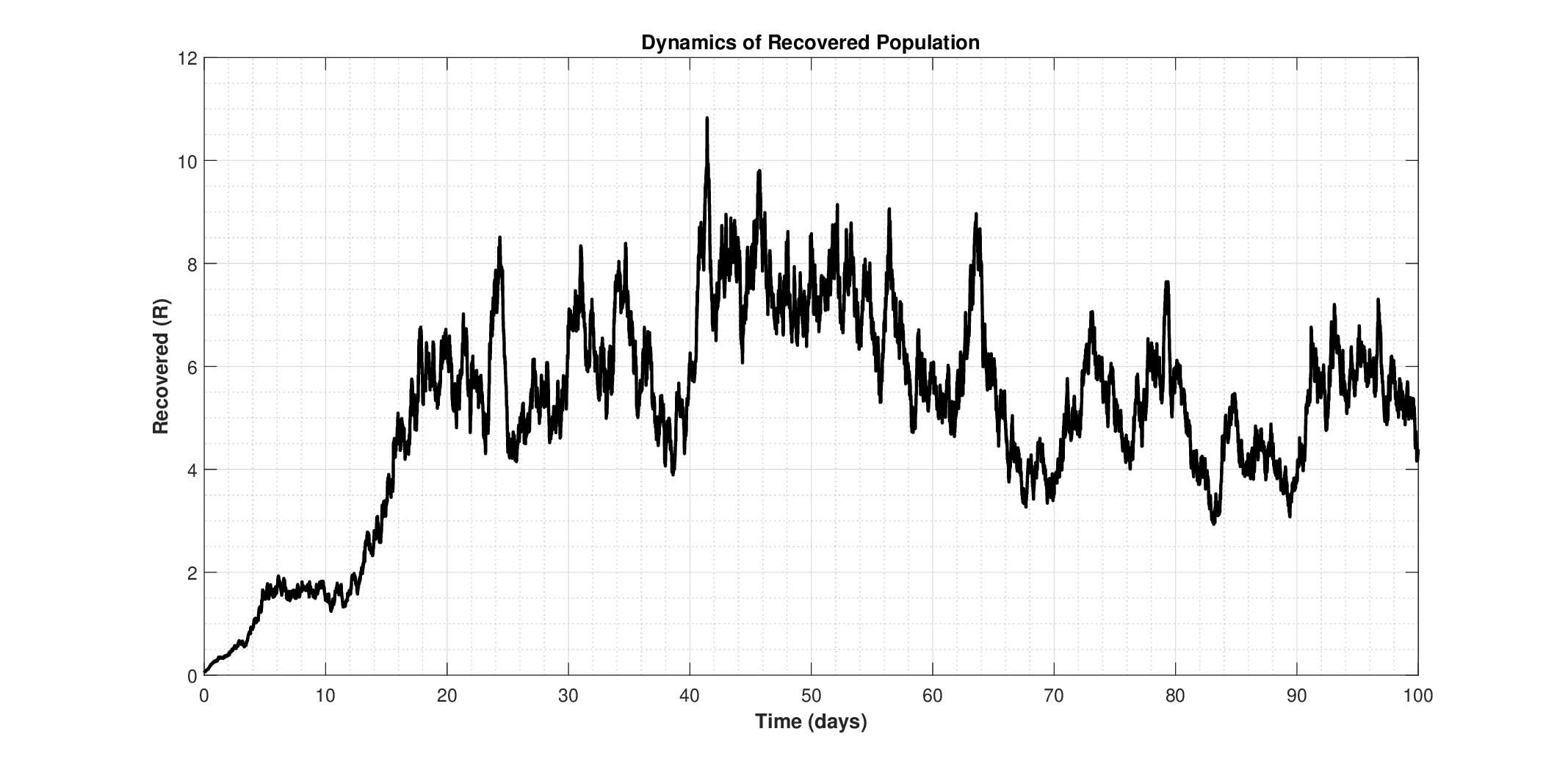}
        \caption{Dynamics of the Recovered (R) \\ population over time.}
        \label{fig:recovered2}
    \end{minipage}
\end{figure}

\section{Conclusion}
In this paper, we built upon the model presented in \cite{global} by integrating two types of noise: white and Lévy noise. We first established the existence of a global positive solution for the extended model. Subsequently, we examined specific conditions under which the disease would become extinct. We then investigated the circumstances that would lead to the persistence of the disease, identifying a sufficient condition for this outcome.

{\small


\begin{thebibliography}{99}
\bibitem{1}
B. Boukanjime, T. Caraballo, M. El Fatini, M. El Khalifi, Dynamics of a stochastic coronavirus (COVID-19) epidemic model with Markovian switching, Chaos Solitons Fractals 141 (2020) 110–361.
\bibitem{2}
R. Taki, M.E. Fatini, M.E. Khalifi, M. Lakhal, K. Wang, Understanding death risks of COVID-19 under media awareness strategy: a stochastic approach, J. Anal. 30 (1) (2022) 79–99.
\bibitem{3}
Cumsille P, Rojas-Díaz Ó, de Espanés PM, Verdugo-Hernández P. Forecasting
COVID-19 Chile’ second outbreak by a generalized SIR model with constant time
delays and a fitted positivity rate. Math Comput Simul (2022), 193, 1–18.
\bibitem{4}
Khan MA, Atangana A. Modeling the dynamics of novel coronavirus (2019-nCoV) with fractional derivative. Alexandria Eng J 2020.
\bibitem{who}
World Health Organization, WHO coronavirus disease (Covid-19) dashboard, https://covid19.who. int/, December 20 2020.
\bibitem{12}
S. Ansumali, S. Kaushal, A. Kumar, M. K. Prakash, and M. Vidyasagar. Modelling a pandemic with asymptomatic patients, impact of lockdown and herd immunity, with applications to SARS-CoV-2. Annual reviews in control, 2020.
\bibitem{13}
N. I. Stilianakis, A. S. Perelson, and F. G. Hayden. Emergence of drug resistance during an influenza epidemic: insights from a mathematical model. Journal of Infectious Diseases, (1998), 177(4):863–873.
\bibitem{14}
E. J. Nelson, J. B. Harris, J. G. Morris, S. B. Calderwood, and A. Camilli. Cholera transmission: the host, pathogen and
bacteriophage dynamic. Nature Reviews Microbiology, (2009,) 7(10):693–702, 
\bibitem{15}
Mathias Peirlinck, Kevin Linka, Francisco Sahli Costabal, Jay Bhattacharya, Eran Bendavid, John PA Ioannidis, and Ellen Kuhl. Visualizing the invisible: The effect of asymptomatic transmission on the outbreak dynamics of covid-19. Computer
Methods in Applied Mechanics and Engineering, (2020) 372, 113410.

\bibitem{rob}
M. Robinson, N.I. Stilianakis, A model for the emergence of drug resistance in the presence of asymptomatic infections, Math. Biosci. 243 (2) (2013) 163–177.

\bibitem{anso}
S. Ansumali, S. Kaushal, A. Kumar, M.K. Prakash, M. Vidyasagar, Modelling a pandemic with asymptomatic patients, impact of lockdown and herd immunity, with applications to SARS-CoV-2, Annu. Rev. Control (2020).
\bibitem{ottaw} 
S. Ottaviano, M. Sensi, and S. Sottile. Global stability of SAIRS epidemic models. Nonlinear Analysis: Real World Applications, 65:103501, 2022.
\bibitem{global}
A. A. Essak, B. Boukanjime,  Global stability of an SAIRS epidemic model with vaccinations, transient immunity and treatment, Nonlinear Analysis: Real World Applications, 73 (2023) 103887.
\bibitem{5}
M. El Fatini, B. Boukanjime, Stochastic analysis of a two delayed epidemic model incorporating Lévy processes with a general non-linear transmission, Stoch. Anal. Appl. (2019) 1–16.
\bibitem{6}
B. Boukanjime, M. El Fatini, A. Laaribi, R. Taki, K. Wang, A Markovian regime-switching stochastic hybrid time-delayed epidemic model with
vaccination, Automatica (2021) 133 .
\bibitem{7}
A. Laaribi, B. Boukanjime, M. El Khalifi, D. Bouggar, M. El Fatini, A generalized stochastic SIRS epidemic model incorporating mean-reverting Ornstein–Uhlenbeck process. Phys. A 615 (2023) 128-609.

\bibitem{8}
B. Berrhazi, M. El Fatini, A. Laaribi, R. Pettersson, A stochastic viral infection
model driven by Lévy noise, Chaos Solitons Fractals 114 (2018) 446–452.
\bibitem{10}
A. Settati, A. Lahrouz, M. El Jarroudi, M. El Fatini, K. Wang, On the threshold dynamics of the stochastic SIRS epidemic model using adequate
stopping times, Discrete Contin. Dyn. Syst. B 25 (2020) 1985–1997.

\bibitem{doo} 
B. Ivorra, M.R. Ferrández, M. Vela-Pérez, A. Ramos, Mathematical modeling of the spread of the coronavirus disease 2019 (COVID-19) taking into
account the undetected infections. The case of China, Commun. Nonlinear Sci. Numer. Simul. 88 (2020) 105303.

\bibitem{levy}
J. Rosinski, Tempering stable processes, Stoch. Proc. Appl., 117 (2007), 677–707.
\bibitem{ley}
N. Privault, L. Wang, Stochastic SIR Levy jump model with heavy tailed increments, J. Nonlinear Sci., 31 (2021), 15.

\bibitem{38}
Zhao Y, Jiang D. The threshold of a stochastic SIRS epidemic model with saturated incidence. Appl Math Lett 2014;34:90–3.
\bibitem{hep}
B. Boukanjime, M.EL. Fatini, A stochastic Hepatitis B epidemic model driven by Lévy noise, Phys. A 521 (2019) 796–806.

\bibitem{num}
 X.Zou, K.Wang, Numerical simulations and modeling for stochastic biological systems with jumps, Communications in Nonlinear Science and Numerical Simulation 19 (5) (2014) 1557–1568.

%

%
%
\bibitem{nmr}
J. Wu, B. Tang, N. L. Bragazzi, K. Nah, Z. McCarthy, Quantifying the role of social distancing, personal protection
and case detection in mitigating covid-19 outbreak in ontario, canada, Journal of Mathematics in Industry 10 (1) (2020) 1–12.
\bibitem{nmr1}
B. Tang, X. Wang, Q. Li, N. L. Bragazzi, S. Tang, Y. Xiao, J. Wu, Estimation of the transmission risk of the 2019-ncov and its implication for public health interventions, Journal of clinical medicine 9 (2) (2020) 462.
\bibitem{nmr2}
P. H. Ontario, Ontario COVID-19 Data Tool, PHO official website.


\bibitem{28}
 P. Roy, S. Jain, M. Maama, Assessing the viability of tri-trophic food chain model in designing a conservation plan: The case of dwindling quokka population, Ecological Complexity
 41 (2020) 100811.
 
\bibitem{29}
 P. Roy, S. Jain, M. Maama, The role of Allee effect in cannibalistic species: An action plan to sustain the declining cod population, Mathematical Modelling of Natural Phenomena
 19 (2024) 15.

 \end{thebibliography}
\end{document}